\documentclass[final,3p]{elsarticle}
\usepackage{amssymb,amsthm,amsmath}
\usepackage{yhmath}
\usepackage{lineno}
\usepackage[colorlinks=true,breaklinks=true,pdftex]{hyperref}
\usepackage[english]{babel}
\usepackage[utf8]{inputenc}
\usepackage{indentfirst}
\usepackage{cases}
\usepackage{color}
\usepackage{graphicx}
\usepackage{caption,subcaption}
\usepackage{fancyhdr}
\usepackage{mathtools}
\usepackage{rotating}
\usepackage{listings}
\usepackage[toc,page]{appendix}
\usepackage{pdfpages}
\usepackage{verbatim}
\usepackage{mathrsfs}
\usepackage{enumitem}
\usepackage[dvipsnames]{xcolor}
\usepackage{geometry}
\usepackage{tikz}
\usepackage{float}
\usepackage{amsbsy}
\usepackage{float} 
\usepackage{cleveref}
\usepackage{lscape}
\usepackage{multirow}
\usepackage{stmaryrd}
\usepackage{natbib}
\usepackage{graphicx}
\usepackage{bm}
\usepackage{extarrows}
\usepackage[percent]{overpic}

\usetikzlibrary{positioning}
\usetikzlibrary{calc}

 \theoremstyle{plain}
 \newtheorem{theorem}{Theorem}[section]
 \newtheorem{lemma}[theorem]{Lemma}
 \newtheorem{proposition}[theorem]{Proposition}
 \newtheorem{corollary}[theorem]{Corollary}

 \theoremstyle{definition}
 \newtheorem{definition}[]{Definition}
 \newtheorem{example}{Example}

\newcommand{\mbf}[1]{\bm{#1}}

\newcommand{\ovl}[1]{\overline{#1}}

\newcommand{\N}{\mathbb N}

\newcommand{\C}{\mathbb C}

\newcommand{\R}{\mathbb R}

\renewcommand{\mbf}{\mathbf}
\renewcommand{\P}{\mathbb{P}}
\newcommand{\bs}{\boldsymbol}

\newcommand{\bP}{\bm P}

\newcommand{\bQ}{\bm Q}
\newcommand{\bC}{\bm C}

\newcommand{\cN}{\mathcal{N}}
\newcommand{\cbU}{\boldsymbol{\mathcal{U}}}
\newcommand{\cU}{\mathcal{U}}

\tikzset{minimum size=2em} 

\modulolinenumbers[1]

\let\today\relax
\makeatletter
\def\ps@pprintTitle{%
	\let\@oddhead\@empty
	\let\@evenhead\@empty
	\def\@oddfoot{}%
	\let\@evenfoot\@oddfoot}
\makeatother

\newcommand\mm[1]{{\color{green}{#1}}}
\newcommand{\pablito}{\textcolor{blue}}

\journal{Elsevier}

\theoremstyle{definition}
\newtheorem*{remark}{Remark}

\begin{document}
	
\begin{frontmatter}

\title{Explicit Inversion of Planar NURBS Curves}

\author[1]{Michelangelo Marsala\corref{cor1}}
\ead{michelangelo.marsala@unifi.it}
\author[2]{Pablo Maz\'on}

\cortext[cor1]{Corresponding author}
\address[1]{Department of Mathematics and Informatics ``U. Dini'', University of Florence, Florence, Italy}
\address[2]{Department of Mathematics, CUNEF University, Madrid, Spain}
\date{\today}

	\begin{abstract}
    We prove that a general planar NURBS curve parametrization $\phi: [u_0,u_m] \xrightarrow{} C \subset \R^2$ admits an inverse map $\phi^{-1}: C \xrightarrow{} [u_0,u_m]$ defined by rational splines. 
    More specifically, we construct a family of rational spline functions on the curve $C$, present explicit formulas for their computation, and prove that the inverse parametrization admits a representation as a linear combination of these functions. 
    Several examples are provided to illustrate the effectiveness of the proposed approach.
	\end{abstract}

    \end{frontmatter}

    \section{Introduction}
    
Non-Uniform Rational B-Splines (NURBS) are among the most widely used representations for curves and surfaces in computational geometry, providing a flexible and mathematically rigorous framework for geometric modeling. Due to their versatility, NURBS have become a cornerstone in Computer-Aided Geometric Design (CAGD), Computer-Aided Design and Manufacturing (CAD/CAM), and other industrial applications. 
They allow designers to model complex geometries with high precision while maintaining smoothness and continuity properties essential for manufacturing and engineering analysis. 
Furthermore, NURBS are increasingly employed in simulation frameworks, particularly in isogeometric analysis, where the same basis functions are used for both geometry representation and the numerical solution of partial differential equations, bridging the traditional gap between design and analysis. 

 By construction, the computation of NURBS objects rely on fast and efficient algorithms for B-spline function evaluation, of which one of the most used is Cox-De Boor formula \cite{deboor,Piegl1997}. 
Other algorithms that directly evaluate B-splines can be found in \cite{Beccari2022,Cohen1980}, while others exploit their piecewise rational nature, as in \cite{Chudy2023,Romani2004} and references therein. Such NURBS constructions find application in many fields, such as topology optimization problems \cite{Gao2023,Giele2025,Zheng2020}, point cloud fitting \cite{Barazzetti2015,Chouychai2015,Dimitrov2014}, geometric design \cite{Bichet2025,Lin2007,Marsala2024} and isogeometric analysis \cite{Bracco2023,Bracco2024,Collin2016,Farahat2024,Marsala2024_iga,Takacs2023,Toshniwal2017,Weinmller2022}.

If the parametric domain of the problem is not $[0,1]^n$, as most often in practice and where multipatch constructions are needed, the evaluation of basis functions or geometry is performed involving the pullback of the NURBS map that defines a patch.  
Such evaluation of the inverse map is usually numerical, as a root finding problem that involves Newton-like methods as in \cite{Krishnamurthy2008,Ma2003,Selimovic2006}.
Since such objects are widely used in many applied fields, an explicit symbolic formula for the inverse map of a NURBS curve or a tensor-product surface or volume would significantly reduce the cost of computing preimages numerically, as the problem would then reduce to simple function evaluation.

Surprisingly, although spline parametrizations are piecewise rational, algebraic properties such as the existence of an inverse spline are largely absent in current NURBS technologies. 
The bijectivity of a parametrization is a crucial requirement for most applications \cite{volumetric_spline_parametrization} and inherently ensures the existence of an inverse map. 
However, this inverse is generally not defined by spline functions, making its practical derivation challenging due to the complex relationships between points and their preimages \cite{birational_quadrilateral}. 
Much research has focused on developing sufficient criteria for local injectivity, aiming to determine conditions under which a rational parametrization is injective within a specific domain \cite{locally_injective}, but this approach falls short when it comes to efficiently computing preimages. 

To date, only the invertibility of rational maps, known as birational, has been studied, with no attention given to spline parametrizations. 
Birational geometry, a classical field in algebraic geometry, only saw practical use in CAD design starting in 2015 \cite{birational_quadrilateral}. 
In dimension two, current approaches are limited to birational low-degree bivariate tensor-product parametrizations \cite{birational_quadrilateral,birational_1xn} and some quadratic cases \cite{birational_quadratic}. 
Regarding birational volumetric parametrizations, the situation is more challenging and only recently constructive results for birational volumes with sufficient flexibility for CAD have been proposed \cite{birational_trilinear_tensor}. 

\section*{Contributions}

   In this work, we study the invertibility of planar NURBS curve parametrizations. 
    We prove that such parametrizations are generically invertible for every degree, and derive explicit formulas for an inverse parametrization relying on rational spline functions defined over the curve. 
    Our aim is that our formulas are useful in applications while addressing the task of computing computing preimages. 
    Specifically, given a NURBS parametrization $\phi:[u_0,u_m] \xrightarrow{} C\subset \R^2$ of a plane curve $C$ the objectives of this paper are:  
\begin{enumerate}
    \item To define a family of rational spline functions over the curve $C\subset \R^2$, that we call physical rational splines. 
    \item To prove that, if $\phi$ is general (\Cref{eq: general}), there exists an inverse map $\phi^{-1}:C \xrightarrow{} [u_0,u_m]$ defined by physical rational splines, i$.$e$.$ the inverse is also a rational spline parametrization. 
    \item To provide an explicit formula for $\phi^{-1}$ as a linear combination of physical rational splines.  
    
\end{enumerate}

\noindent Such construction is a first step towards the definition of explicit inverse map of tensor-product NURBS objects, that would save many computational cost in application such as isogeometric analysis simulations and point cloud fitting problems.

 \section*{Outline}
 The paper is divided into five sections. In \Cref{preliminaries} we recall basic definitions and properties of splines and NURBS. \Cref{loc_inv_section} introduces the main tools we use to define the inverse parametrization of a planar NURBS curve, that is used in \Cref{phys_rat_splines} to define rational spline functions represented in physical coordinates. Such physical rational splines are hence used in \Cref{sec: inverse B-spline} to give an alternative definition of the inverse NURBS parametrization. \Cref{sec_examples} concludes the paper with various numerical examples illustrating the potential of the proposed construction.

\section{Preliminaries}\label{preliminaries}

In this paper, a \textit{curve} is the image of a non-constant  continuous map $\gamma: [u_0,u_m] \xrightarrow{} \R^N$, where $u_m > u_0$. 
The following is a specialization of the general definition of rational spline functions to curves (see \cite{the_algebra_of_splines}).

\begin{definition}
    \label{def: spline on curve}
    Let $C\subset \R^N$ be a curve. 
    Moreover, let $U = (u_0,u_1,\ldots,u_m)$ be a sequence of $m+1$ points in $C$ and let $S = (s_0,\ldots,s_m)\in\N^{m+1}$. 
    A \textit{rational spline of smoothness $S$ at $U$} is a function $f:C\xrightarrow{} \R$ satisfying the following two conditions: 
    \begin{enumerate}
        \item[1)] The restriction of $f$ to each connected component of $C \, \backslash \, \{ u_0,\ldots,u_m \}$ is a rational function. 
        \item[2)] $f$ is smooth at $u_k$ up to order $s_k$, for each $0\leq k\leq m$.  
    \end{enumerate}
    If for every connected component in $1)$ the spline $f$ restricts to a rational function where the degree of both numerator and denominator is $\leq d$, we say that the \textit{degree of $f$} is bounded by $d$.  
\end{definition}

\noindent 
If $N=1$ the situation is classical since curves are closed intervals $[u_0,u_1]\subset \R$. 
In this case, spline functions are best understood through B-splines, which we recall briefly. 
Given $m\in \N$ let 
\begin{equation*}
U= \{ u_0, u_1,\dots, u_{m} \}
\end{equation*}
be a sequence of $m+1$ nondecreasing real numbers. We refer to it as \textit{knot vector} and to its elements as \textit{knots}. 
For each $0\leq k\leq m -1$, we call $I_k=[u_k,u_{k+1})$ a \textit{knot interval}. 
Given $m\geq d + 1$, we can define the B-splines of degree $d$ with knot vector $U$. The $k$-th B-spline of degree $d$ is defined as 
\begin{equation}
\label{bspline_def}
N_{k,d}(u)=\frac{u-u_k}{u_{k+d}-u_k} N_{k,d-1}(u) + \frac{u_{k+d+1}-u}{u_{k+d+1}-u_{k+1}} N_{k+1,d-1}(u)
\end{equation}
where
\begin{equation*}
N_{k,0}(u)= 
\begin{cases*}
1 \quad &\text{if } $u \in I_k$ \\
0 \quad &\text{otherwise}
\end{cases*} 
\end{equation*}
It is straightforward from the definition that $N_{k,d}(u)$ restricts to a polynomial over any knot interval $I_k$. 
The recursive formula \eqref{bspline_def} is known as  \textit{Cox-de Boor algorithm} or \textit{Cox-de Boor formula}. 
For identical consecutive knots, namely $u_k = u_{k+1} = \ldots = u_{k+s}$ for some $s>0$, some of the denominators in \eqref{bspline_def} vanish. 
In these cases, we adopt the usual convention and set $0/0 = 0$. Lastly, the multiplicity of a knot $u_k$ is the number of its occurrences in $U$, denoted by $\mu(u_k)$. 

B-splines enjoy properties such as nonnegativity, local support, and they form a  partition of unitity. Remarkably, B-splines form a basis for the vector space of spline functions of degree $\leq d$ on $[u_0, u_m]$ with smoothness $m_k$ at each knot $u_k$, which motivates their name \textit{Basis splines}. For an exhaustive description of these and other properties, we refer the reader to \cite{deboor,Piegl1997,Prautzsch2002}. 

A NURBS (Non-Uniform Rational B-Spline) curve $C$ of degree $d$ is defined by 
\begin{align}
\label{nurbs_curve}
    \phi: [u_0,u_m] & \longrightarrow C \subset \R^N
    \\[2pt] \nonumber
    &u \longmapsto \dfrac{\sum_{i=0}^n w_i \mbf{P}_i N_{i,d}(u)}{\sum_{i=0}^n w_i N_{i,d}(u)} 
    = 
    \left(
    \dfrac{f_1(u)}{f_0(u)}
    ,\ldots,
    \dfrac{f_n(u)}{f_0(u)}
    \right)
\end{align}
where $\{\bm{P}_i\}_{i=0}^n$ are the control points (forming a \textit{control polygon}), the $\{w_i\}_{i=0}^n$ are the \textit{weights}, and where $f_0,\ldots, f_N$ are the defining polynomials of the map. The number of control points, knots and the degree of the curve are related by $ n = m - d - 1$.
Without loss of generality, throughout the paper we work with knot vectors of the form 
\begin{equation*}\label{knot_vector_endpoints}
U=\{\underbrace{0,\dots,0}_{d+1},u_{d+1},\dots,u_{m-d-1},\underbrace{1,\dots, 1}_{d+1} \}
\end{equation*}
where the multiplicity of inner knots is possible greater than one. 
This condition ensures that $C$ interpolates the endpoints of the control polygon, i.e. $\phi(0)=\mbf{P}_0$ and $\phi(1)=\mbf{P}_n$. 
Further properties of such objects are presented, e.g., in~\cite{Piegl1997,Prautzsch2002}. 
Moreover, we define the \textit{reduced knot vector} 
\begin{equation*}
U' 
= 
\{ 
u_0',
\ldots,
u_{m'}'
\}
\end{equation*}
as the nondecreasing sequence with the  knots from $U$ but all having  multiplicity exactly one.

Throughout the paper we will only deal with planar curves.
For $N=2$, the control points take the form $\mbf{P}_{i} = (x_{i},y_{i})$, and we can write the defining polynomials of $\phi$ as 
\begin{equation*}
\Phi(u) 
= 
(f_0(u),f_1(u),f_2(u))
= 
\sum_{i=0}^n
w_i\ovl{\mbf{P}}_i N_{i,d}(u)
\end{equation*}
where we set $\ovl{\mbf{P}}_i = (1,x_i,y_i)$ for each $0\leq i\leq n$. 
Clearly, from $\Phi(u)$ we recover the parametrization $\phi$. 

\section{Computation of the Local Inverses}\label{loc_inv_section}
        By definition, the curve $C$ parametrized by \eqref{nurbs_curve} is piecewise rational. 
        In the sequel, we will assume the following genericity condition. 
        \begin{definition}
        \label{eq: general}
        Let $\phi: [u_0,u_m] \xrightarrow{} C\subset \R^2$ be as \eqref{nurbs_curve} with B-splines of degree $d$. 
        We say that $\phi$ is \textit{general} if no rational segment of $C$ can be parametrized with polynomials of degree $<d$. 
        \end{definition}
        For any $d \geq 1$, $\phi$ is general unless there is a special algebraic relation among the control points. 

        \begin{example}
        If $d=2$, the parametrization $\phi$ is general if and only if no three consecutive control points are collinear. 
        Namely, for any consecutive control points $\mbf{P}_j$, $\mbf{P}_{j+1}$, $\mbf{P}_{j+2}$ we can always find a knot interval $I_k = [u_k,u_{k+1}]$ where the only nonzero B-splines are precisely $N_{j,2}(u)$, $N_{j+1,2}(u)$, $N_{j+2,2}(u)$. 
        In particular, the rational piece $\phi(I_k)$ is a segment if and only if the three control points are collinear. 
        \end{example}

    \subsection{Conversion from NURBS to B\'ezier Representation} 
    \label{subsection: conversion}
    In order to compute explicitly the inverse NURBS parametrization, we need
    to extract rational representations of these functions on each $I_k$. 
    In this subsection, we rely on the work \cite{Yan2024} to perform the conversion of a NURBS curve into its local Bézier representation on $I_k$.

    \vskip4pt
    
    For any $k=0,\ldots, m-1$ there are at most $d+1$ B-spline functions $N_{j,d}(u)$ that are nonzero on the knot interval $I_k$, namely $N_{k-d,d}(u),\ldots,N_{k,d}(u)$. 
    Moreover, the restriction $F_{j,d}(u) = N_{k-d+j,d}(u)|_{I_k}$ is a polynomial of degree $d$ for any $j = 0, \ldots, d$. 
    Therefore, there is a $(d+1)\times (d+1)$ change-of-basis  matrix $S_{k,d}$ satisfying 
    \begin{equation*}
        \begin{pmatrix}
            F_{0,d}(t)
            & 
            \ldots 
            & 
            F_{d,d}(t)
        \end{pmatrix}
        = 
        \begin{pmatrix}
            B_{0,d}(t)
            & 
            \ldots 
            & 
            B_{d,d}(t)
        \end{pmatrix}
        \cdot 
        S_{k,d}
    \end{equation*}
    where we are using the change of variable $\alpha: I_k \xrightarrow{} [0,1)$ given by 
    \begin{equation}
    \label{eq: linear change of coordinates}    
    t = \alpha(u) = \frac{u-u_k}{u_{k+1}-u_k} \in [0,1)
    \end{equation}
    and $B_{j,d}(t) = \binom{d}{j} (1-t)^{d-j} u^{j}$ is the $i$-th Bernstein polynomial of degree $d$. 
    In \cite{Yan2024}, symbolic formulas and algorithms for the computation of the change-of-basis matrix $S_{k,d}$ are provided. 
    In particular, we have 
    \begin{equation}
    \label{eq: conversion}
        \Phi(t)|_{[0,1)} 
        =
        \begin{pmatrix}
            B_{0,d}(t)
            & 
            \ldots 
            & 
            B_{d,d}(t)
        \end{pmatrix}
        \cdot 
        S_{k,d}
        \cdot 
        \begin{pmatrix}
            w_{i-d} \ovl{\mbf{P}}_{i-d} \\
            \vdots \\
            w_i \ovl{\mbf{P}}_i
        \end{pmatrix} 
        = 
        (
        f_0^{(k)}(t)
        , 
        f_1^{(k)}(t)
        , 
        f_2^{(k)}(t)
        )
    \end{equation}
    for some polynomials 
    $
    f_0^{(k)}
    , 
    f_1^{(k)}
    , 
    f_2^{(k)}
    $
    of degree $\leq d$ in $\R[t]$. Here, the superindex $k$ refers to the knot interval $I_k$.
    
    \subsection{Local Rational Inverses}
    \label{loc_inv_sec}
    We define a \textit{physical knot interval} as $C_k := \phi(I_k)$. 
Here we rely on the local B\'ezier representation in \eqref{eq: conversion}.
Our approach for the computation of  $\phi^{-1}|_{C_k}$ follows \cite[Section 2]{buse_dandrea}. 
First, after the change of coordinates \eqref{eq: linear change of coordinates} the restriction $\phi_k = \phi|_{I_k}:I_k \xrightarrow{} C_k$ becomes

\begin{align}
\label{nurbs_curve local}
    \psi_k:=\phi_k \circ \alpha^{-1} : [0,1) &\longrightarrow C_k\subset \R^2 \\[2pt] \nonumber
    t &\longmapsto  \left(\dfrac{f_1^{(k)}(t)}{f_0^{(k)}(t)},\dfrac{f_2^{(k)}(t)}{f_0^{(k)}(t)} \right)
\end{align}

Next, we recall the definition of the Sylvester matrix of two univariate polynomials.
\begin{definition}
Let 
$
f(t) = a_m t^m + a_{m-1} t^{m-1} + \dots + a_0
$ 
and 
$
g(t) = b_n t^n + b_{n-1} t^{n-1} + \dots + b_0
$ 
two univariate polynomials in $R[t]$ for some ring $R$.  
The \textit{Sylvester matrix} 
of $f$ and $g$ is the matrix 
\begin{equation*}
\text{Sylv}(f,g) =
\begin{pmatrix}
a_0 & 0 & \cdots & 0 & b_0 & 0 & \cdots & 0 \\
a_1 & a_0 & \cdots & 0 & b_1 & b_0 & \cdots & 0 \\
\vdots & \vdots & \ddots & \vdots & \vdots & \vdots & \ddots & \vdots \\
a_m & a_{m-1} & \cdots & a_0 & b_n & b_{n-1} & \cdots & b_0 \\
0 & a_m & \cdots & a_1 & 0 & b_n & \cdots & b_1 \\
\vdots & \vdots & \ddots & \vdots & \vdots & \vdots & \ddots & \vdots \\
0 & 0 & \cdots & a_m & 0 & 0 & \cdots & b_n
\end{pmatrix}
\in R^{(m+n)\times (m + n)}
\end{equation*}
\end{definition}

Now, for every $k$, we define the polynomials 
$$X(t) = f_1^{(k)}(t) - x f_0^{(k)}(t), \quad Y(t) = f_2^{(k)}(t) - y f_0^{(k)}(t) \in \R[x,y][t]$$
and let $\text{Sylv}(X(t),Y(t))$ be the Sylvester matrix of $X(t)$ and $Y(t)$ with respect to the variable $t$. 
Namely, $\text{Sylv}(X(t),Y(t))$ is a $2d\times 2d$ matrix where the first $d$ columns have entries in $\R[x]$ and the last $d$ columns have entries in $\R[y]$. 
Observe that 
$$
\text{rank } \text{Sylv}(X(t),Y(t)) < 2d \quad \iff \quad (x,y) = \psi_k(t)
$$
for some $t\in \C$. 
Moreover, we have the identity 
$$
\begin{pmatrix}
    t^{2d-1} &
    \ldots &
    t & 
    1
\end{pmatrix}
\cdot 
\text{Sylv}(X(t),Y(t))
= 
\begin{pmatrix}
    t^{d-1} X(t) &
    \ldots  &
    t X(t) &
    X(t) &
    t^{d-1} Y(t) &
    \ldots  &
    t Y(t) &
    Y(t)
\end{pmatrix}
$$
Hence, the vector $(t^{2d-1}, \ldots , t , 1)$ lies in the (left) kernel of $\text{Sylv}(X(t),Y(t))$ if and only if $X(t) = Y(t) = 0$, or equivalently, if $\psi_k(t) = (x,y)$. 
Finally, let $M_i$ be the signed $(2d-1)\times (2d-1)$ minor of $\text{Sylv}(X(t),Y(t))$ obtained by erasing the $i$-th row and any fixed column. 
Then, by \cite[Proposition $2.1$]{buse_dandrea} given any $(x,y)\in C_k$ its preimage $t = \phi^{-1}(x,y)$ can be computed as $t = M_i / M_{i+1}$, independently of the choice of $1\leq i\leq 2d$.
In particular, the inverse parametrization $\psi_k^{-1}$ is locally defined on $C_k$ as the rational map 
\begin{align}
\label{nurbs_curve_local}
    \psi^{-1}_k : C_k &\longrightarrow [0,1) \\[2pt] \nonumber
    (x,y) &\longmapsto  \frac{M_i(x,y)}{M_{i+1}(x,y)}
\end{align}

The following result asserts that $\phi$ admits a global piecewise rational inverse $\phi^{-1}$. 

\begin{theorem}
\label{theorem: piecewise NURBS inverse}
If $\phi$ is general, it admits the inverse $\phi^{-1} : C \xrightarrow{} [u_0,u_m]$ given by 

\begin{align}
\label{inv_with_locals}
    (x,y) &\longmapsto  \phi^{-1}(x,y)=
        \begin{cases*}
        (u_{1}-u_0) \, \psi^{-1}_0(x,y) + u_0 \quad &\text{if } $(x,y)\in C_0$ \\
        \hspace{2cm}\vdots \\
        (u_{k+1}-u_k) \, \psi^{-1}_k(x,y) + u_k \quad &\text{if } $(x,y)\in C_k$ \\
        \hspace{2cm}\vdots \\
        (u_m-u_{m-1}) \, \psi^{-1}_{m-1}(x,y) + u_{m-1} \quad &\text{if } $(x,y)\in C_{m-1}$ \\
    \end{cases*}
\end{align}
\end{theorem}

\begin{proof}
    It is enough to prove the formula locally on $I_k = [u_k,u_{k+1})$ and $C_k = \phi(I_k)$, for any $0\leq k\leq m-1$. 
    First, from the change of variable 
    $
    t = \alpha(u) = (u - u_k) / (u_{k+1} - u_k)
    $ 
    we obtain the inverse change of variable $u = \alpha^{-1}(t) = (u_{k+1}-u_k) t + u_k$. 
    By definition, we have the identities 
    $$
    \phi_k(u) = 
    \phi_k(\alpha^{-1}(t))
    = 
    (\phi_k \circ \alpha^{-1})(t)
    = 
    \psi_k(t)
    =
    \psi_k(\alpha(u))
    = 
    (\psi_k \circ \alpha )(u)
    = 
    \phi_k(u) 
    \ .
    $$
    Hence, $\psi_k = \phi_k \circ \alpha^{-1}$ over $[0,1)$ and $\phi_k = \psi_k \circ \alpha$ over $I_k$. 
    Moreover, by \eqref{nurbs_curve local} and \eqref{nurbs_curve_local} the composition $\psi_{k}^{-1} \circ \psi_k = \text{id}_{[0,1)}$ is the identity on $[0,1)$. 
    From \eqref{inv_with_locals}, we have 
    $$
    \phi^{-1}(x,y) = (u_{k+1} - u_k) \, \psi_k^{-1}(x,y) + u_k = (\alpha^{-1}\circ \psi_k)(x,y)
    $$  
    for $(x,y)\in C_k$. Therefore, we can write 
    $$
    \phi^{-1}
    \circ 
    \phi
    = 
    (\alpha^{-1} \circ \psi_k^{-1}) \circ 
    (\psi_k \circ \alpha)
    = 
    \alpha^{-1}
    \circ 
    (
    \psi_k^{-1}
    \circ 
    \psi_k
    )
    \circ 
    \alpha
    =
    \alpha^{-1}
    \circ 
    \alpha
    = 
    \text{id}_{I_k}
    $$
    for every $u\in I_k$. Similarly we deduce that $\phi \circ \phi^{-1} = \text{id}_{[0,1)}$. 
\end{proof} 

The formulas for the inverse map $\phi^{-1}$ from Theorem \ref{theorem: piecewise NURBS inverse} are valid for splines of any degree $d \geq 1$. 
However, in the quadratic case $d = 2$ one can give a more explicit characterization of the relationship between the defining polynomials of $\phi^{-1}$ and the associated control points. 
Specifically, for each $k = 0,\ldots, m-1$, following \cite{Yan2024}, we define 
$$
a_k = \frac{u_k - u_{k-1}}{u_{k+1} - u_{k-1}} 
\ ,\ 
b_k = \frac{u_{k+1} - u_{k}}{u_{k+2} - u_{k}} 
\ ,\ 
S_{k,2}[w] = 
\begin{pmatrix}
    w_{k-2}( 1-a_k) & w_{k-2} a_k & 0 \\
    0   & w_{k-1} & 0 \\ 
    0   & w_k (1-b_k) & w_k b_k
\end{pmatrix}
\  .
$$
Here, the matrix $S_{k,2}[w]$ coincides with the matrix $S_{k,2}$ from Section \ref{subsection: conversion} after multiplying the $i$-th row by the weight $w_{k- 2 + i}$ for each $0\leq i\leq 2$. 
In particular, if all the weights are equal to one, the two matrices coincide. 
The proof of the following lemma is based on syzygies, and is independent of the general approach adopted in this section. 

\begin{lemma}
Let $\phi$ be quadratic, without three consecutive control points in a line. 
Then, $\phi$ admits an inverse rational spline $\phi^{-1}: C \xrightarrow{} [u_0,u_m]$, defined for $(x,y)\in C_k$ as 
$$
\phi^{-1}
(x,y)
= 
\frac
{
\frac{1}{2}
\det\left( {S_{k,2}[w]}^{(2)} \right)}
{
\frac{1}{2}
\det\left( {S_{k,2}[w]}^{(2)} \right)
- 
\det\left( {S_{k,2}[w]}^{(1)} \right)
}
= 
\frac
{ \det\left( {S_{k,2}[w]}^{(3)} \right) }
{
\det\left( {S_{k,2}[w]}^{(3)} \right)
- 
\frac{1}{2}
\det\left( {S_{k,2}[w]}^{(2)} \right)
}
$$
where ${S_{k,2}[w]}^{(i)}$ is obtained by replacing the $i$-th column of ${S_{k,2}[w]}$ with $\bs{\rho}(x,y)$, and  
$$
\bs{\rho}(x,y) 
= 
(
\rho_0(x,y)
,
-\rho_1(x,y)
,
\rho_2(x,y)
)^t
\ \ ,\ \ 
\rho_i(x,y) 
= 
\det
\left(  
\mbf{P}_j
, 
\mbf{P}_{k}
, \mbf{X}
\right)
\ \ 
, 
\ \ 
\mbf{X} = (1,x,y)^t
$$
for each $\{ i,j,k \} = \{ 1,2,3 \}$ where $j<k$. 
\end{lemma}


\section{Physical Rational Spline Functions}\label{phys_rat_splines}
    Recall that $U'$ is the reduced knot vector of $U$. 
    For each knot $u_j' \in U'$, we define the corresponding \textit{physical knot} as $\mbf{U}_j' = \phi(u_j')$. 
    Given multiplicities $\mu(\mbf{U}_{j}')\geq 0$, we define the \textit{phyisical knot vector} as 
    \begin{equation*}
    \mathcal{U}
    =
    \{ 
    \mbf{U}_0,
    \ldots,
    \mbf{U}_M
    \}
    =
    \{
    \underbrace{\mbf{U}_0',\ldots,\mbf{U}_0'}_{\mu(\mbf{U}_0')} 
    , 
    \ldots
    ,
    \underbrace{
    \mbf{U}_{j}',\ldots,\mbf{U}_{j}'}_{ \mu(\mbf{U}_{j}')} 
    , 
    \ldots 
    ,
    \underbrace{
    \mbf{U}_{m'}',\ldots,\mbf{U}_{m'}'}_{\mu(\mbf{U}_{m'}')}
    \}
    \ .
\end{equation*}
    In this section, we define a family of rational splines supported on $C$ relative to $\mathcal{U}$.
    We call these functions \textit{physical rational} splines, since they are defined by coordinates in the physical domain.
    \begin{definition}[Physical Rational Spline]
    \label{def: physical B-spline}
    Let $C \subset \R^2$ be the curve of degree $d$ defined by $\phi$ as \eqref{nurbs_curve} and $\phi^{-1}$ be as \eqref{inv_with_locals}.
    Given $p,k$ satisfying $0\leq k,p \leq M -1$, the $k$-th physical rational spline $\mathcal{N}_{k,p}:C\xrightarrow{}\R$  with physical knot vector $\mathcal{U}$ is the function recursively defined as 
    \begin{align}
    \label{phys_coxdeBoor}
\mathcal{N}_{k,p}(x,y) & = \frac{
\phi^{-1}(x,y) - \phi^{-1}(\mbf{U}_k)
}{
\phi^{-1}(\mbf{U}_{k+p})-\phi^{-1}(\mbf{U}_k)
} 
\,
\mathcal{N}_{k,p-1}(x,y) 
+ 
\frac{
\phi^{-1}(\mbf{U}_{k+p+1})-\phi^{-1}(x,y)
} 
{
\phi^{-1}(\mbf{U}_{k+p+1}) - \phi^{-1}(\mbf{U}_{k+1})
}
\,
\mathcal{N}_{k+1,p-1}(x,y),
\end{align}
where
\begin{equation*}
\mathcal{N}_{k,0}(x,y)= 
\begin{cases*}
1 \quad &\text{if} $(x,y)\in C_{k}$ \\
0 \quad &\text{otherwise}
\end{cases*}.
\end{equation*} 
\end{definition}

\begin{remark}
    The degree $p$ in Definition \ref{def: physical B-spline} is not necessarily equal to the degree of the defining B-splines of $\phi$.
\end{remark}

As in \eqref{bspline_def}, if the quotients in \eqref{phys_coxdeBoor} yield a division $0/0$, we set it to be 0. 
Here, the recursive formula \eqref{phys_coxdeBoor} is the analog of the Cox-de Boor formula \eqref{bspline_def} over the curve $C$. 

The following result motivates the introduction of physical rational splines. 
In particular, it establishes that such functions are the pullbacks by $\phi^{-1}$ of the parametric B-splines, whenever the multiplicities of the physical and parametric knots coincide. 

\begin{lemma}
\label{lemma: physical B-splines as pullbacks}
    Let 
    $
    U = \{ u_0,\ldots,u_m \} 
    $ and  
    $
    \mathcal{U} = \{\mbf{U}_0 ,\ldots, \mbf{U}_m\}
    $ 
    be such that $\phi(u_k) = \mbf{U}_k$, i$.$e$.$ corresponding knots on the parametric and physical knot vectors have exactly the same multiplicity. 
    Then, we have  
    \begin{equation}
        \label{eq: physical B-splines as pullbacks}
        \mathcal{N}_{k,p}(x,y)
        =
        N_{k,p}(\phi^{-1}(x,y))
        \ ,\ 
        N_{k,p}(u)
        = 
        \mathcal{N}_{k,p}(\phi(u))
    \end{equation}
    for every $(x,y)\in C$ and $u\in [u_0,u_m]$.
\end{lemma}

\begin{proof}
By Theorem \ref{theorem: piecewise NURBS inverse}, it is straightforward that 
\begin{equation*}
    \mathcal{N}_{k,0}(x,y)
    = 
    \begin{cases*}
        1 \quad &\text{if} $(x,y)\in C_{k}$ \\
        0 \quad &\text{otherwise}
    \end{cases*}
    = 
    \begin{cases*}
        1 \quad &\text{if} $\phi^{-1}(x,y)\in I_{k}$ \\
        0 \quad &\text{otherwise}
    \end{cases*}
    = 
    N_{k,0}(\phi^{-1}(x,y))
\end{equation*}
Similarly, we obtain $N_{k,0}(u) = \mathcal{N}_{k,0}(\phi(u))$. 
Hence, in order to prove \eqref{eq: physical B-splines as pullbacks} we can proceed by induction on the degree $p$ with $p = 0$ as the base case. 

Assume that the result is true for degree $p-1$, where $0 \leq p-1 < m-1$. 
We next prove that \eqref{eq: physical B-splines as pullbacks} also holds for degree $p$. 
Given $0\leq k\leq m-1$, the $k$-th spline of degree $p$ on $[u_0,u_m]$ is
$$
N_{k,p}(u)
=
\frac{u-u_k}{u_{k+p}-u_k} N_{k,p-1}(u) + \frac{u_{k+p+1}-u}{u_{k+p+1}-u_{k+1}} N_{k+1,p-1}(u)
\ .
$$
Substituting $u = \phi^{-1}(x,y)$ we obtain 
\begin{align*}
N_{k,p}(\phi^{-1}(x,y))
& = 
\frac{\phi^{-1}(x,y)-u_k}{u_{k+p}-u_k} N_{k,p-1}(\phi^{-1}(x,y)) + \frac{u_{k+p+1}-\phi^{-1}(x,y)}{u_{k+p+1}-u_{k+1}} N_{k+1,p-1}(\phi^{-1}(x,y)) \\ 
& = 
\frac{
\phi^{-1}(x,y) - u_k
}{
u_{k+p}-u_k
} 
\,
\mathcal{N}_{k,p-1}(x,y) 
+ 
\frac{
u_{k+p+1}-\phi^{-1}(x,y)
} 
{
u_{k+p+1} - u_{k+1}
}
\,
\mathcal{N}_{k+1,p-1}(x,y) 
\end{align*}
where in the second equality we use the induction hypothesis. 
Finally, as $u_k = \phi^{-1}(\mbf{U}_k)$ we can write 
\begin{align*}
N_{k,p}(\phi^{-1}(x,y))
& =  
\frac{
\phi^{-1}(x,y) - \phi^{-1}(\mbf{U}_k)
}{
\phi^{-1}(\mbf{U}_{k+p})-\phi^{-1}(\mbf{U}_k)
} 
\,
\mathcal{N}_{k,p-1}(x,y) 
+ 
\frac{
\phi^{-1}(\mbf{U}_{k+p+1})-\phi^{-1}(x,y)
} 
{
\phi^{-1}(\mbf{U}_{k+p+1}) - \phi^{-1}(\mbf{U}_{k+1})
}
\,
\mathcal{N}_{k+1,p-1}(x,y) 
\\ 
& = \mathcal{N}_{k,p}(x,y) 
\end{align*}
that is the first identity in \eqref{eq: physical B-splines as pullbacks}. 
Since $\phi^{-1}(\phi(u)) = u$ for every $u\in [u_0,u_m]$, we obtain 
$$
N_{k,p}(\phi^{-1}(\phi(u)))
= 
N_{k,p}(u)
=
\mathcal{N}_{k,p}(\phi(u))
$$
and the statement follows. 
\end{proof}

Lemma \ref{lemma: physical B-splines as pullbacks} ensures that physical rational splines inherit the analogous properties of parametric B-splines.
To make this precise, we briefly recall the definition of regularity classes for parametric curves.

\begin{definition}
    \label{definition: regularity on curves}
    Let $\phi: I\subset \R \xrightarrow{} \R^N$ 
    be a parametrization of a curve, and let $F:\phi(I) \xrightarrow{} \R$ be a function. 
    We say: 
    \begin{equation*}
        F\in C^n(\phi(I)) 
        \iff 
        F \circ \phi \in C^n(I)
    \end{equation*}
\end{definition}

Physical rational spline functions therefore satisfy the same standard properties as parametric B-spline functions.

\begin{proposition}\label{prop_phys_bsplines}
With the previous notation, the following properties hold: 
    \begin{enumerate}
        \item (Nonnegativity) $\mathcal{N}_{k,p}(x,y) \geq 0$ for every $(x,y)\in C$.
        \item (Local Support) 
        $\mathcal{N}_{k,p}(x,y) = 0$ if $(x,y)\not\in C_k \cup \ldots \cup C_{k+p}$.  
        \item (Partition of Unity) $\sum_{k = 0}^{M-p-1} \mathcal{N}_{k,p}(x,y) = 1$.
        \item (Regularity) $\mathcal{N}_{k,p}(x,y)\in C^{p-\mu(\mbf{U}_k)}(\phi((u_{k-1},u_{k+1})))$ 
        for every $0<k<m$.
        Equivalently, $\mathcal{N}_{k,p}(x,y)$ is smooth up to order $p - \mu(\mbf{U}_k)$ at the physical knot $\mbf{U}_k$.
    \end{enumerate}
    \begin{proof}
    		The proof follows immediately from the standard properties of B-splines and Lemma \ref{lemma: physical B-splines as pullbacks}.
    	
    \end{proof}
    \end{proposition}

\section{Rational Spline Representation of the Inverse} 

\label{sec: inverse B-spline}

In this section, we prove that a general NURBS curve parametrization admits a rational spline representation. 
Namely, we write the inverse parametrization  globally as a linear combination of physical rational splines. 
Additionally, we provide an explicit formula for such an inverse inverse. 


\begin{definition}[Greville points]
Let $U = \{ u_0, \ldots, u_{m} \}$ be a knot vector. 
For each $0\leq i\leq m - d - 1$ the \textit{Greville points} associated with the $B$-spline $N_{i,d}(u)$ is  
\begin{equation*}
\label{eq: Greville abscissa}
\xi_{i,d} = \frac{1}{d} \sum_{j=1}^d u_{i+j}.
\end{equation*}  
\end{definition}

Since B-splines reproduce polynomials up to their degree, for every $u\in[u_0,u_m]$ and $1 \leq d \leq m-1$ we have the identity
\begin{equation}
\label{eq: identity Greville}
u = \sum_{i=0}^{n} \xi_{i,d} N_{i,d} (u) \ .
\end{equation}
The identity above means that B-splines have linear precision. 
Finally, we have the following result. 
\begin{theorem}
\label{thm_inv_spline}
Let $U = \{u_0,\ldots,u_m\}$ be a knot vector and $\mathcal{U} = \{\phi(u_0), \ldots, \phi(u_{m})\}$ be its associated physical knot vector.
Then, if a general parametrization $\phi:[u_0,u_m]\xrightarrow{} C$ as in \eqref{nurbs_curve} is injective, it admits the inverse $\phi^{-1}:C\xrightarrow{} [u_0,u_m]$ given by 
\begin{equation*}\label{inv_Bspline_nurbs}     
            \phi^{-1}(x,y)=\sum_{i= 0}^{n} \xi_{i,p}\, \cN_{i,p}(x,y)
\end{equation*} 
for every $ 1 \leq p \leq m-1$, 
where $\cN_{i,p}$ are the physical rational splines defined on $\mathcal{U}$. 
\end{theorem}

\begin{proof}
Since the map is injective, it is enough to check that $(\phi^{-1} \circ \phi)(u) = u$ for every $u\in[u_0,u_m]$. 
Using Lemma \ref{lemma: physical B-splines as pullbacks} we can write 
\begin{gather*}
(\phi^{-1}\circ \phi)(u)
= 
\phi^{-1}(\phi(u))
= 
\sum_{i= 0}^{n} \xi_{i,p}\, \cN_{i,p}(\phi(u))
= 
\sum_{i= 0}^{n} \xi_{i,p}\, N_{i,p}(u)
= 
u
\end{gather*}
where the last identity follows from \eqref{eq: identity Greville}. 
The statement follows.  
\end{proof}

\begin{remark}
    The requirement that $\phi$ be injective fails only when the curve $C$ has self-intersections. 
    Indeed, since the restriction $\phi|_{I_k}: I_k \xrightarrow{} C_k$ is birational for each knot interval $I_k$, the parametrization is locally one-to-one on every $I_k$. 
    If a point $\mbf{Q}\in C$ belongs to both $\phi(I_k)$ and $\phi(I_{k'})$ for $0\leq k < k' < m$, that is, if the curve has a self-intersection, then the local inverses provided by Theorem \ref{theorem: piecewise NURBS inverse} assign distinct preimages to $\mbf{Q}$. 
    Consequently, $\phi^{-1}$ is no longer a single-valued map but a multivalued one.  
    In Example \ref{sec: example - self-intersection}, we illustrate that $\phi^{-1}$ can nevertheless be used effectively to compute the inverse outside the self-intersections of $C$. 
\end{remark}

\begin{remark}
    It is also possible to define the inverse parametrization by mean of physical rational splines of degree greater than $m-1$. To do that, it is required to adjust the multiplicity of the endpoints of the physical knot vector $\cU$ according to the desired degree $p$.
\end{remark}


\begin{corollary}
    If $\phi$ is injective, then 
    $\phi^{-1}\in C^\infty(C)$. 
    In words, $\phi^{-1}$ is infinitely differentiable on the curve $C$. 
\end{corollary}

\begin{proof}
    It is straightforward from Definition \ref{definition: regularity on curves} since $(\phi^{-1}\circ \phi) = \text{id} \in C^\infty([u_0,u_m])$. 
\end{proof}

    \section{Examples}\label{sec_examples}
    In this section we present different examples in which we apply the proposed construction to compute the inverse function of a given NURBS curve. In particular, we show the case of a quadratic and cubic NURBS curve, a quartic curve with multiple inner knots and a degree 5 NURBS with a self intersection.

      \subsection{Quadratic NURBS}
Let us consider the points
    \begin{equation*}\label{pts_nurbs_quad}
    \mbf{P}_0 = \left(0,0\right),\,  \mbf{P}_1 = \left(\dfrac{1}{4},1\right), \, \mbf{P}_2 = \left(\dfrac{9}{10},-\dfrac{1}{5}\right),\, \mbf{P}_3 = \left(1,\dfrac{2}{5}\right),
    \end{equation*}
    with weights
    \begin{equation*}
    w_0=1,\, w_1=3, \, w_2  = \dfrac{3}{2}, \, w_3 = 1,
    \end{equation*}

     defining a quadratic NURBS $\phi$ on the knot vector $U = \left\{0,0,0,\dfrac{1}{2},1,1,1\right\}$ as
     
     \begin{equation}\label{quad_nurbs_expl}
     \phi(u) = \sum_{i=0}^3 w_i \mbf{P}_i N_{i,2}(u)=(x(u),y(u))=
      \begin{cases*}
            \left(\dfrac{9u^2 - 15u}{55u^2 - 40u - 5},\dfrac{93u^2 - 60u}{ 55u^2 - 40u - 5} \right) \quad &\text{if} $u\in I_{0}$ \\
            \left(\dfrac{-13u^2+19u-1}{5(u-2)^2},\dfrac{47u^2 -80u +35}{5(u-2)^2} \right) \quad &\text{if} $u\in I_{1}$
        \end{cases*}. \\        
     \end{equation}

     \begin{figure}[htbp]
     \centering

  \begin{subfigure}[t]{0.4\textwidth}
    \centering
    \includegraphics[width=\linewidth]{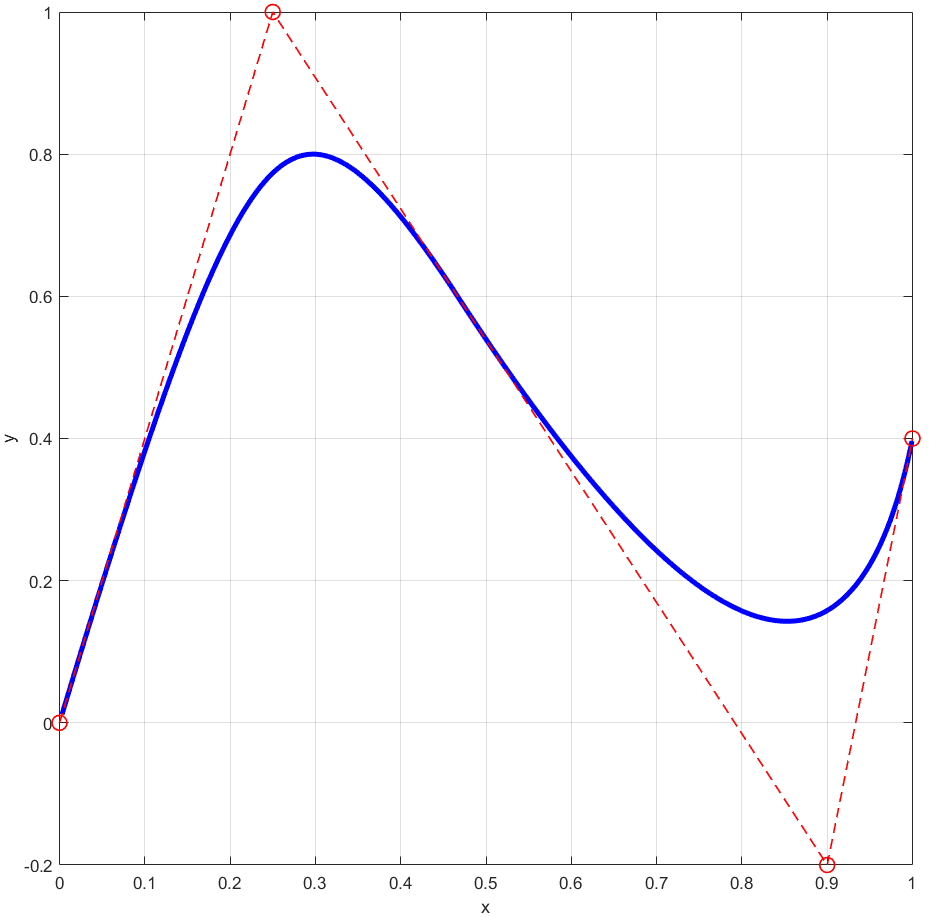}
    \caption{}
  \end{subfigure}
  \begin{subfigure}[t]{0.4\textwidth}
    \centering
    \includegraphics[width=\linewidth]{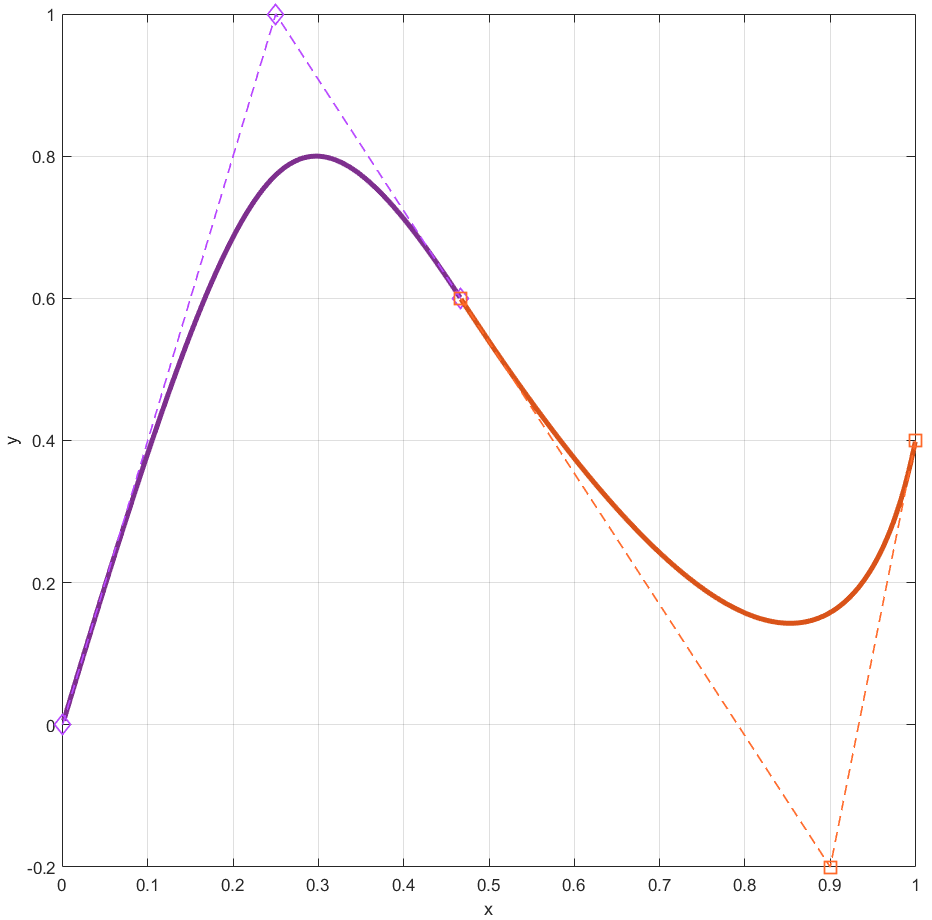}
    \caption{}
  \end{subfigure}
  \caption{NURBS curve with associated control net (a) and its decomposition in rational polynomial pieces (b).}
  \label{nurbs_with_pieces_split_ex}
\end{figure}

Moreover, the physical points associated with the curve are:
    \begin{equation*}
        \mbf{U}_0 = \mbf{P}_0, \,\mbf{U}_1 = \left(\dfrac{7}{15}, \dfrac{3}{5}  \right), \, \mbf{U}_2 = \mbf{P}_3.
    \end{equation*}
    Its graph and its decomposition into (two) rational polynomial pieces are shown in \Cref{nurbs_with_pieces_split_ex}.

    By~\Cref{thm_inv_spline}, the inverse $\phi^{-1}$ can be defined as any linear combination of physical rational splines of a given degree and their associated Greville points. If we take $p=2$, we can define quadratic functions $\cN_{i,2}$ on the physical knot vector 
    $$\cU=\left\{ \mbf{U}_0,\mbf{U}_0,\mbf{U}_0, \mbf{U}_1, \mbf{U}_2, \mbf{U}_2, \mbf{U}_2\right\},$$
    while if we choose $p=1$, the inverse will be defined by linear rational splines $\cN_{i,1}$ with physical knot vector 
    $$\cU=\left\{ \mbf{U}_0,\mbf{U}_0, \mbf{U}_1, \mbf{U}_2, \mbf{U}_2\right\}.$$

    \begin{figure}[htbp]
    \centering
    \begin{subfigure}[b]{0.28\textwidth}
        \centering
        \includegraphics[width=\textwidth]{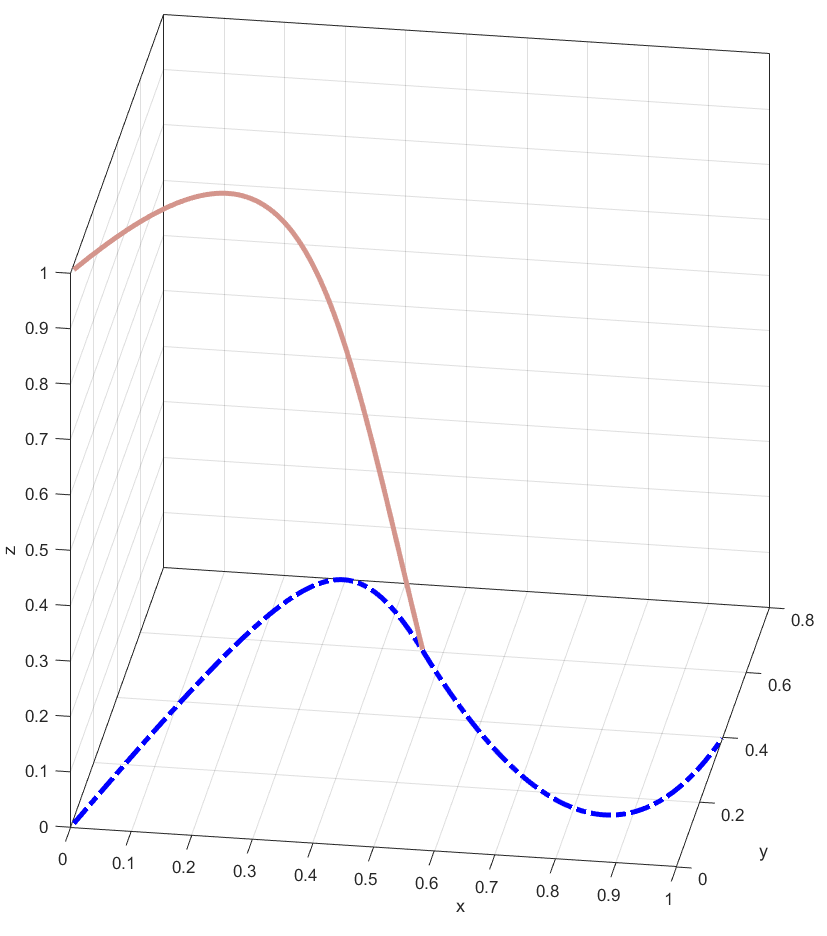}
        \caption{}
    \end{subfigure}
    \hfill
    \begin{subfigure}[b]{0.28\textwidth}
        \centering
        \includegraphics[width=\textwidth]{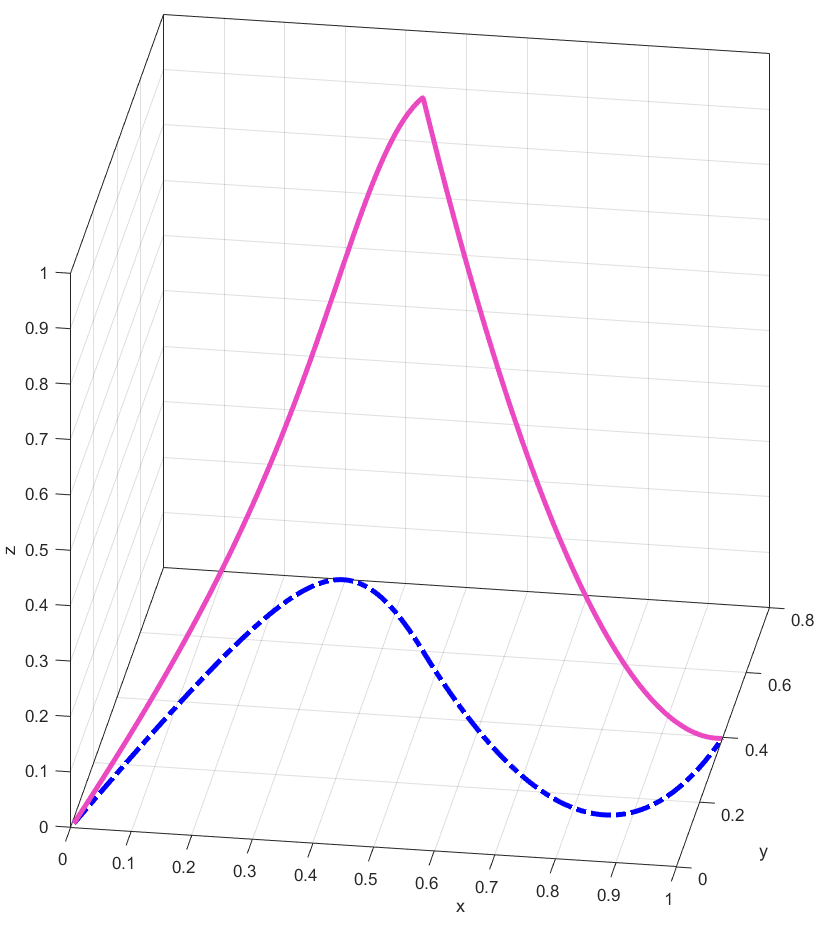}
        \caption{}
    \end{subfigure}
    \hfill
    \begin{subfigure}[b]{0.28\textwidth}
        \centering
        \includegraphics[width=\textwidth]{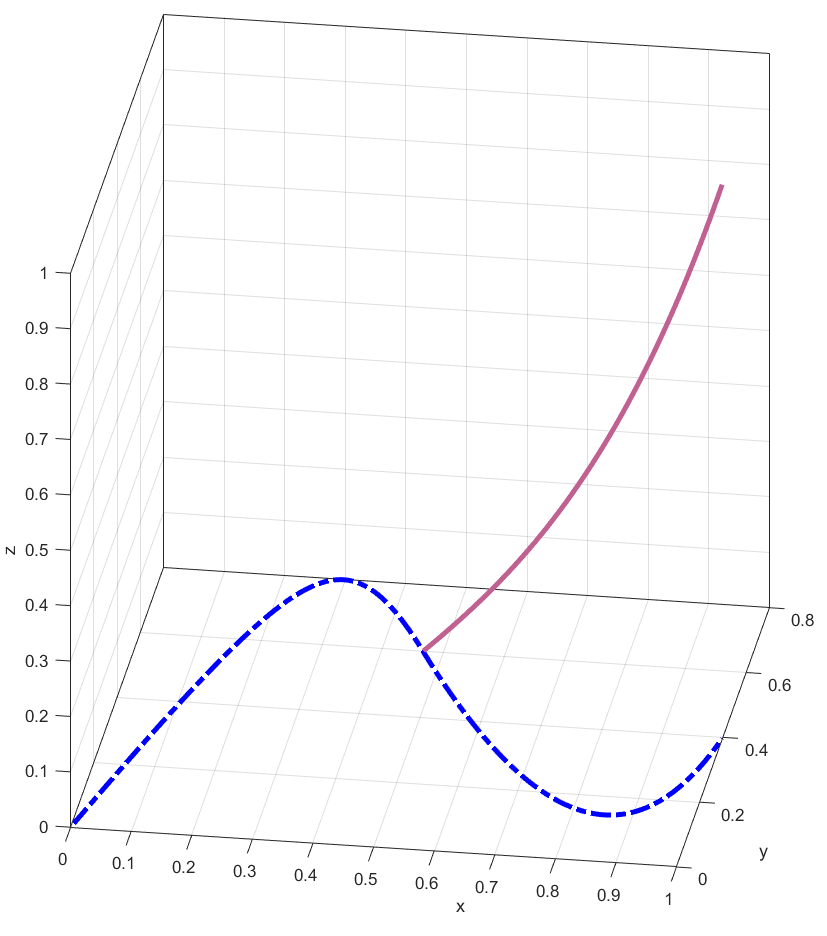}
        \caption{}
    \end{subfigure}

    \vspace{0.5em} 

    \begin{subfigure}[b]{0.33\textwidth}
        \centering
        \includegraphics[width=\textwidth]{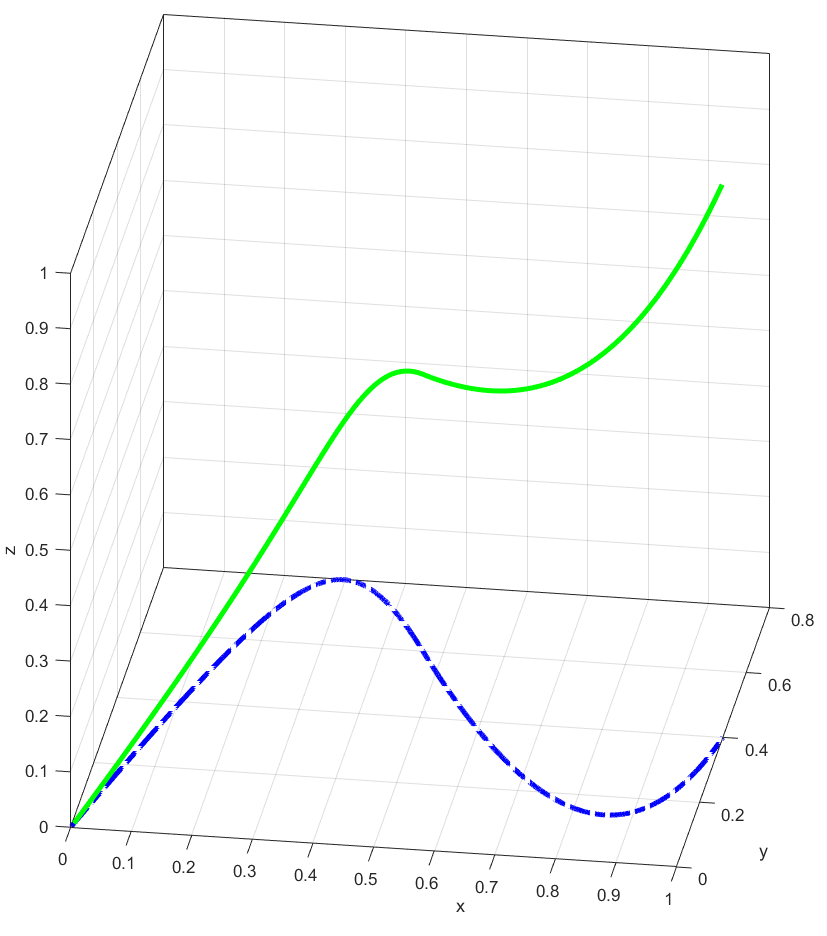}
        \caption{}
    \end{subfigure}

    \caption{The physical B-spline $\cN_{0,1}$ (a), $\cN_{1,1}$ (b), $\cN_{2,1}$ (c) and the graph of the inverse NURBS in \eqref{inv_quad_NURBS}. In dashed line is the input NURBS curve $\phi$.}
    \label{quad_basis_with_inverse}
\end{figure}

    For example, such linear rational splines are defined as:
        \begin{align*}
        &\cN_{0,1}(x,y)=
        \begin{cases*}
            \dfrac{90x+25y-57}{28x+31y-57} \quad &\text{if} $(x,y)\in C_{0}$ \\
          0 \quad &\text{otherwise}
        \end{cases*}, \nonumber \\
         &\cN_{1,1}(x,y)=
        \begin{cases*}
            \dfrac{-62x+6y}{28x+31y-57} \quad &\text{if} $(x,y)\in C_{0}$ \\
          \dfrac{-150x - 60y + 174}{180x + 55y - 49} \quad &\text{if} $(x,y)\in C_{1}$ \\
        \end{cases*}, \nonumber \\     
         &\cN_{2,1}(x,y)=
        \begin{cases*}
            \dfrac{330x + 115y - 223}{180x + 55y - 49} \quad &\text{if} $(x,y)\in C_{1}$ \\
          0 \quad &\text{otherwise}
        \end{cases*}, \nonumber \\
        \end{align*}
whose graphs are shown in \Cref{quad_basis_with_inverse}-(a)-(b)-(c).

In both cases, the inverse parametrization results
    \begin{equation}\label{inv_quad_NURBS}
        \phi^{-1}(x,y) = \sum_{i=0}^3\xi_{i,2}\, \cN_{i,2}(x,y) = \sum_{i=0}^2 \xi_{i,1}\, \cN_{i,1}(x,y) = 
        \begin{cases*}
            \dfrac{ -31x + 3y}{28x + 31y - 57} \quad &\text{if} $(x,y)\in C_{0}$ \\
           \dfrac{255x + 85y - 136}{180x + 55y - 49} \quad &\text{if} $(x,y)\in C_{1}$
        \end{cases*}, \\ 
    \end{equation}
    and \Cref{quad_basis_with_inverse}-(d) illustrates its image. It can be checked that substituting \eqref{quad_nurbs_expl} into \eqref{inv_quad_NURBS} we obtain the identity.

  \subsection{Cubic NURBS}
    The points
    \begin{align*}\label{pts_nurbs}
    \mbf{P}_0 = \left(0,0\right),\,  \mbf{P}_1 = &\left(1,2\right), \, \mbf{P}_2 = \left(\dfrac{3}{4},1\right),\, \mbf{P}_3 = \left(\dfrac{3}{2},-\dfrac{1}{3}\right), \, \mbf{P}_4 = \left(3,\dfrac{7}{4}\right), \, \mbf{P}_5 = \left(\dfrac{11}{4},\dfrac{5}{2}\right), \mbf{P}_6 = \left(4,\dfrac{1}{2}\right),
    \end{align*}
    and weights
    \begin{equation*}
    w_0=1,\, w_1=2, \, w_2 = 1, \, w_3 = \dfrac{5}{2}, \, w_4 = 1, \, w_5=3,\, w_6 = 1,
    \end{equation*}
    define a cubic NURBS curve $\phi$ on the knot vector $U = \left\{0,0,0,0,\dfrac{1}{4},\dfrac{1}{2},\dfrac{3}{4},1,1,1,1\right\}$, whose graph is shown in~\Cref{cubic_and_inverse}-(a). Choosing $p=2$, the physical points forming the physical knot vector 
    $$\cU=\left\{ \mbf{U}_0, \mbf{U}_0, \mbf{U}_0, \mbf{U}_1, \mbf{U}_2, \mbf{U}_3, \mbf{U}_4, \mbf{U}_4, \mbf{U}_4  \right\},$$ 
    on which we can define the quadratic physical rational splines $\cN_{i,2}$ (see~\Cref{Cubic_NURBS_uniform}-(b)) are:
    \begin{equation*}
        \mbf{U}_0 = \mbf{P}_0, \, \mbf{U}_1 = \left(\dfrac{25}{24}, \dfrac{26}{27}  \right), \, \mbf{U}_2 = \left(\dfrac{25}{16}, -\dfrac{7}{144}  \right), \, \mbf{U}_3 = \left(\dfrac{71}{28}, \dfrac{397}{252}  \right), \,\mbf{U}_4 = \mbf{P}_6.
    \end{equation*}
    
 \begin{figure}[htbp]
     \centering
  \begin{subfigure}[t]{0.4\textwidth}
    \centering
    \includegraphics[width=\linewidth]{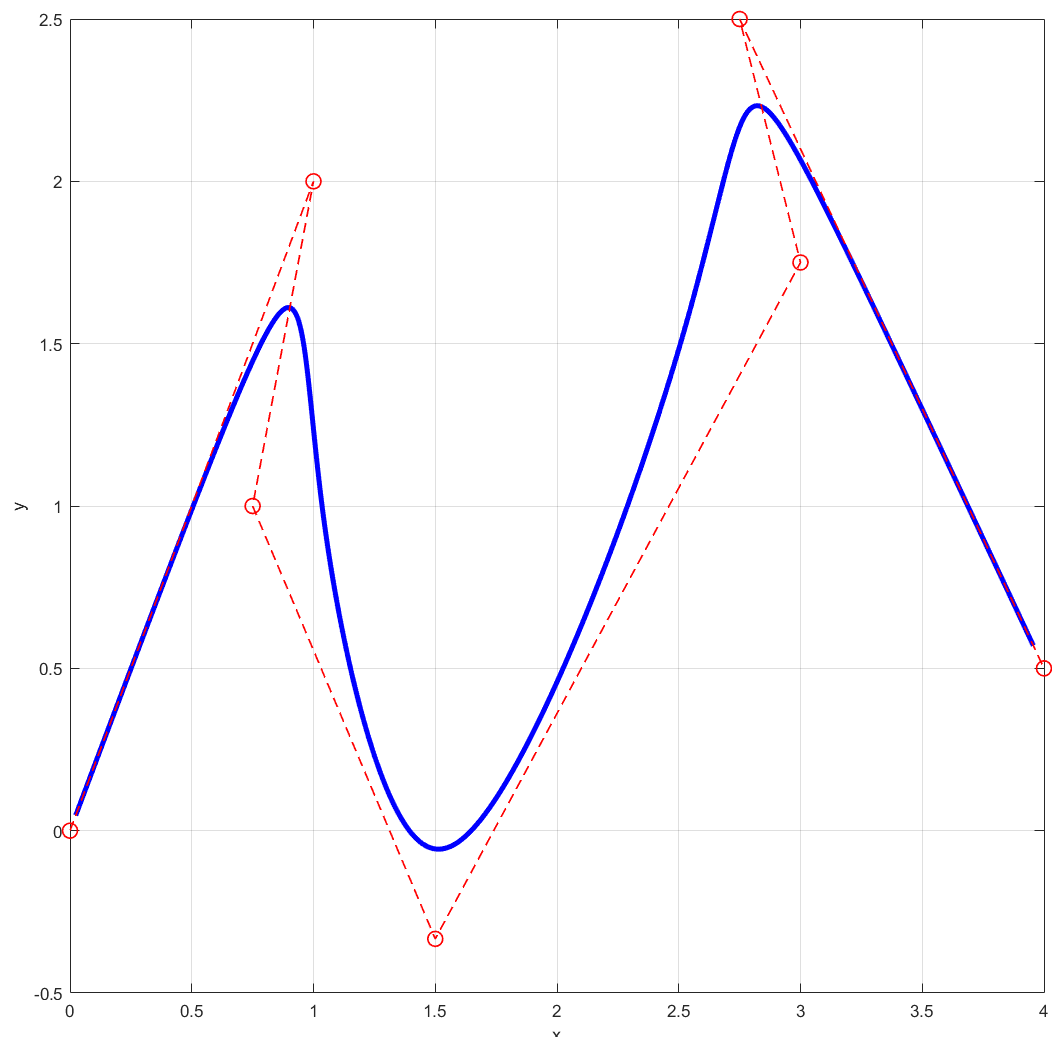}
    \caption{}
  \end{subfigure}
  \begin{subfigure}[t]{0.4\textwidth}
    \centering
    \includegraphics[width=.95\linewidth]{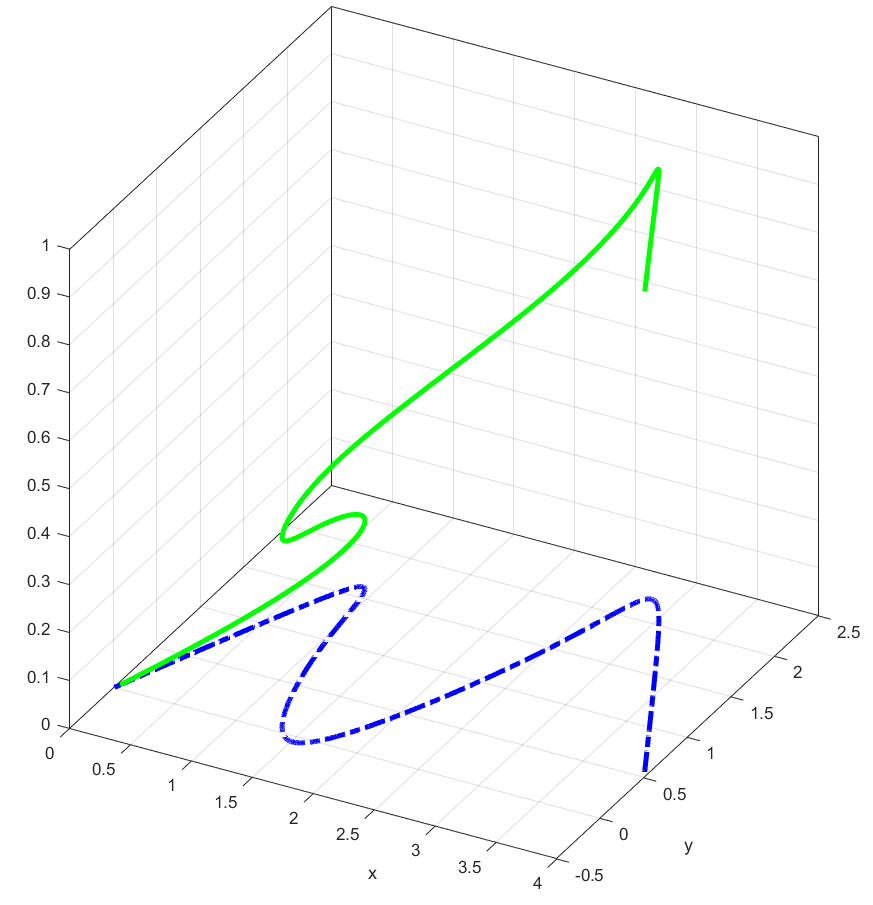}
    \caption{}
  \end{subfigure}
  \caption{NURBS curve with associated control net (a) and the graph of its inverse (b). In dashed line is the input NURBS curve $\phi$.}
  \label{cubic_and_inverse}
\end{figure}

 Hence, the inverse function $\phi^{-1}$ can be defined as
    \begin{equation}\label{inv_cubic}
        \phi^{-1}(x,y) = \sum_{i=0}^5\xi_{i,2}\, \cN_{i,2}(x,y),
    \end{equation}
    whose graph is presented in~\Cref{cubic_and_inverse}-(b).

    \begin{figure}[h!]
    \centering
    \begin{subfigure}[b]{0.28\textwidth}
        \centering
        \includegraphics[width=\textwidth]{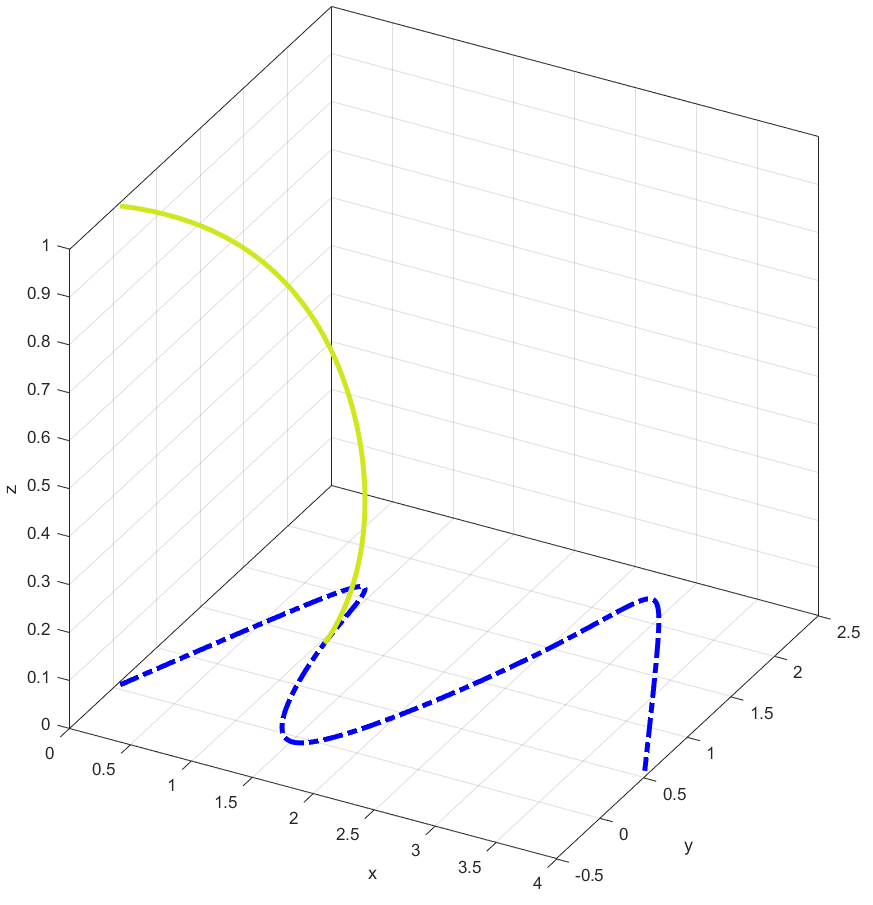}
        \caption{}
    \end{subfigure}
    \hfill
    \begin{subfigure}[b]{0.28\textwidth}
        \centering
        \includegraphics[width=\textwidth]{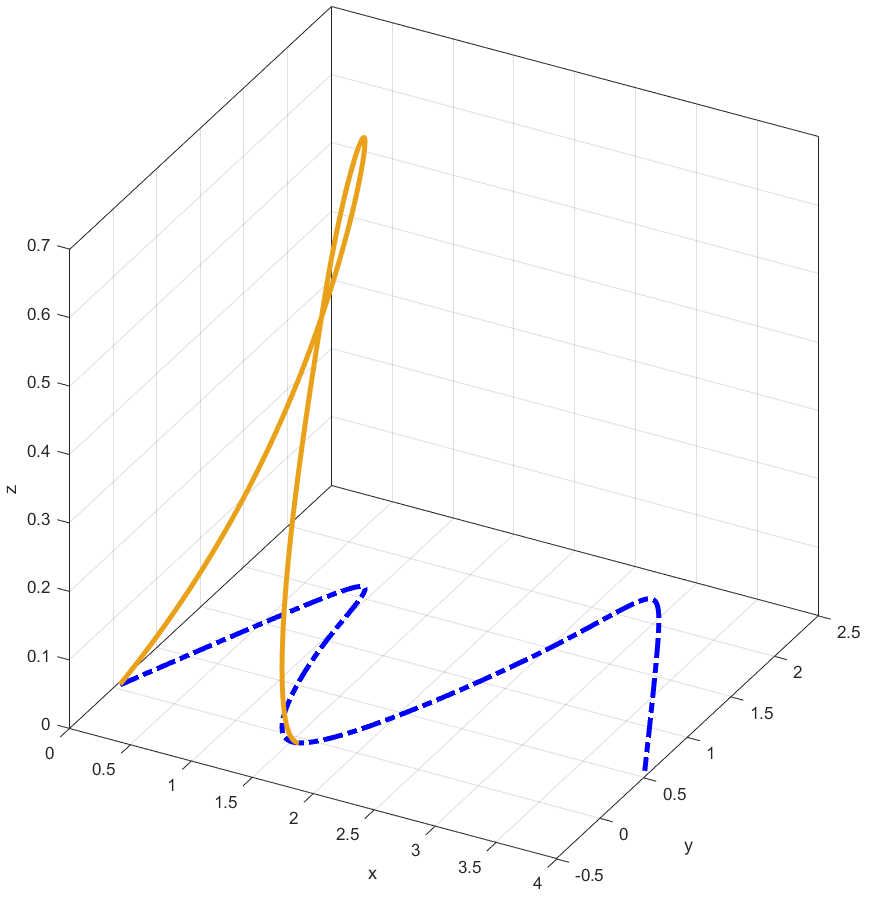}
        \caption{}
    \end{subfigure}
    \hfill
    \begin{subfigure}[b]{0.28\textwidth}
        \centering
        \includegraphics[width=\textwidth]{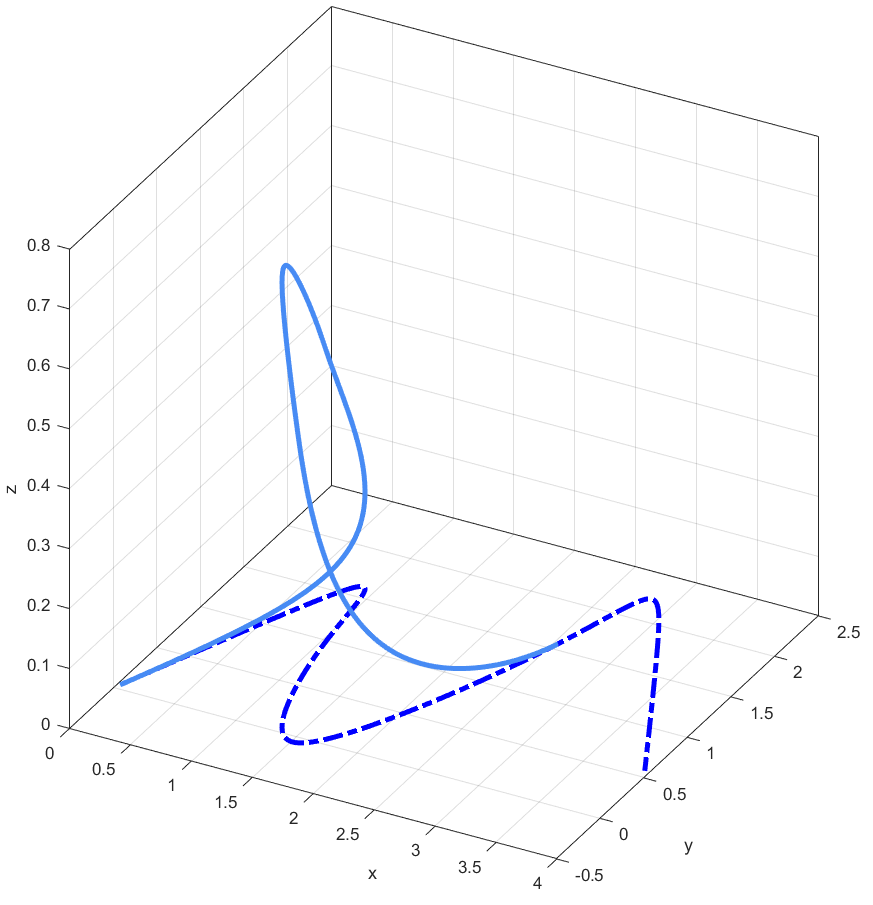}
        \caption{}
    \end{subfigure}

     \begin{subfigure}[b]{0.28\textwidth}
        \centering
        \includegraphics[width=\textwidth]{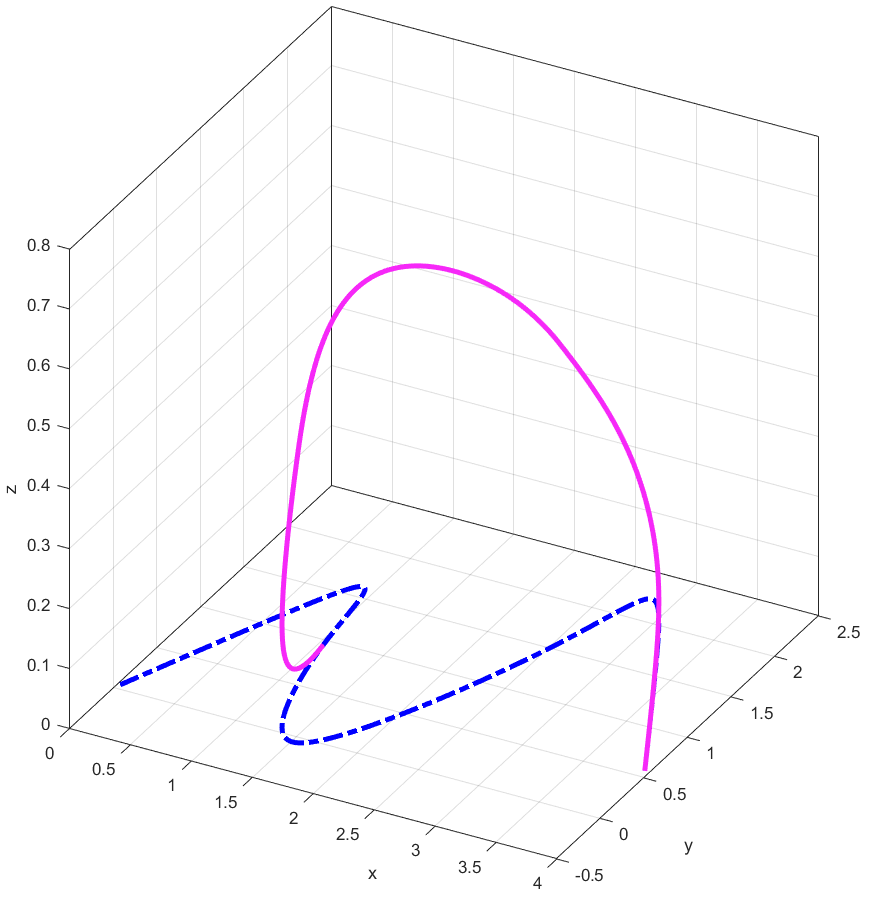}
        \caption{}
    \end{subfigure}
    \hfill
    \begin{subfigure}[b]{0.28\textwidth}
        \centering
        \includegraphics[width=\textwidth]{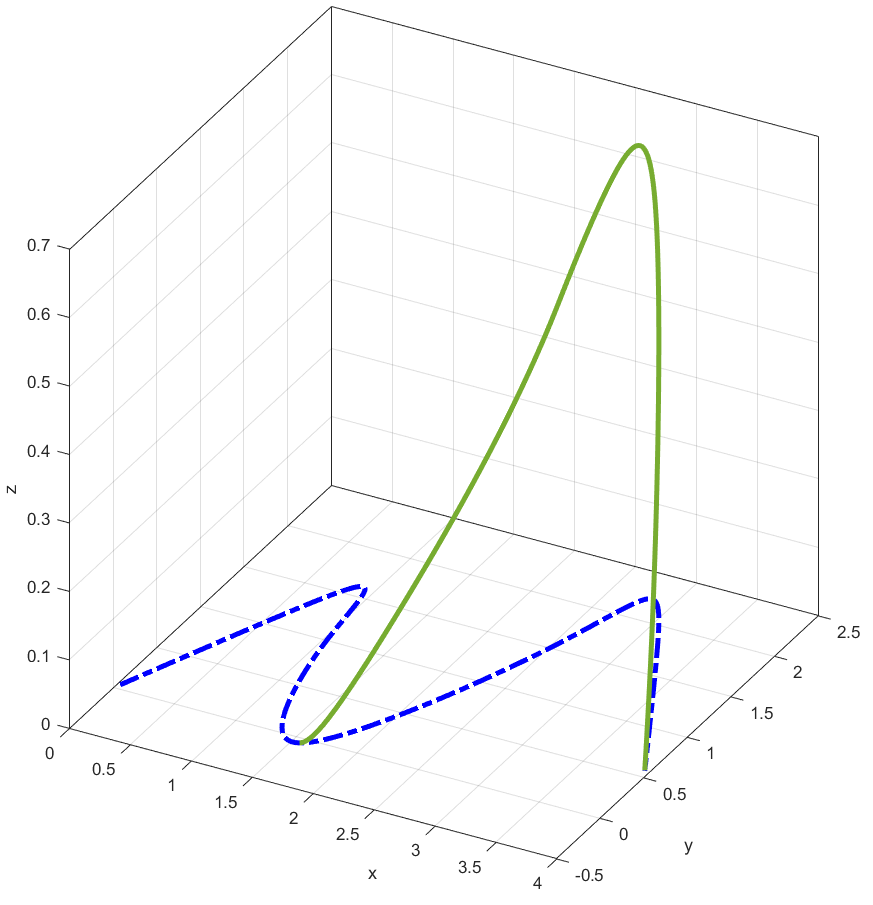}
        \caption{}
    \end{subfigure}
    \hfill
    \begin{subfigure}[b]{0.28\textwidth}
        \centering
        \includegraphics[width=\textwidth]{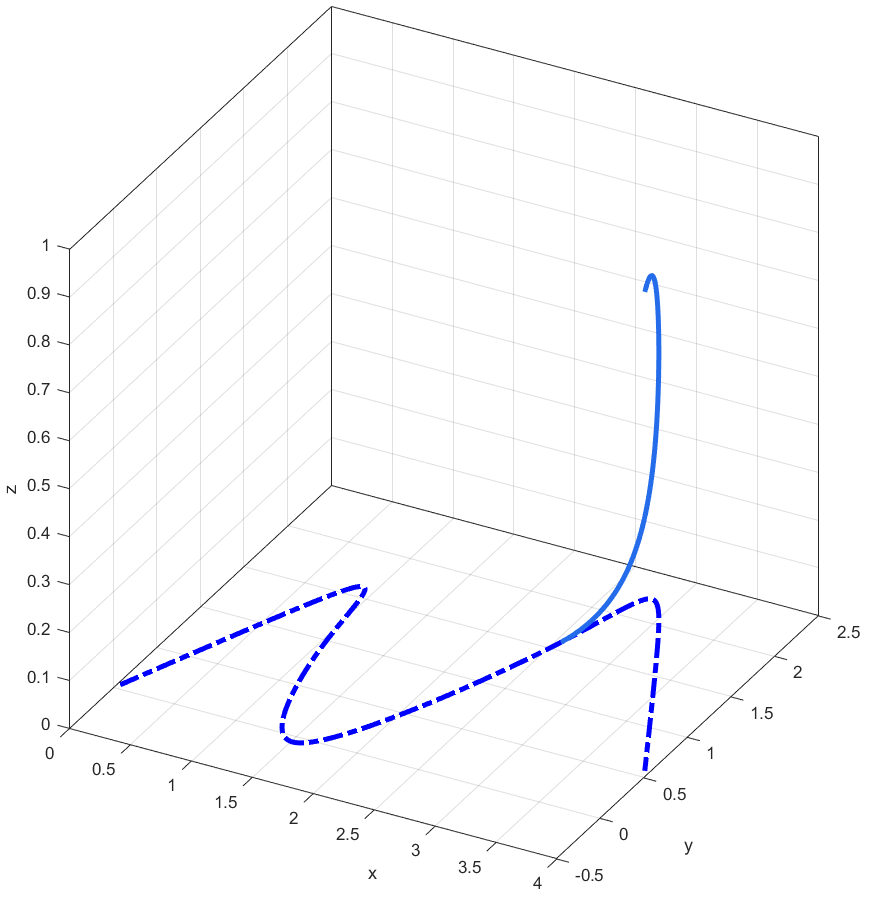}
        \caption{}
    \end{subfigure}

    \caption{From (a) to (f): the quadratic physical rational functions $\cN_{i,2}, i=0,\dots,5$, defining the inverse function in \eqref{inv_cubic}. In dashed line is the input NURBS curve $\phi$.}
    \label{Cubic_NURBS_uniform}
    \end{figure}

    \subsection{Quartic NURBS with multiple inner knot}

    Given the knot vector $U= \left\{0,0,0,0,0,\dfrac{1}{3},\dfrac{2}{3},\dfrac{2}{3},\dfrac{2}{3},\dfrac{2}{3},1,1,1,1,1 \right\}$, the points
    \begin{align*}
    &\mbf{P}_0 = \left(0,0\right),\,  \mbf{P}_1 = \left(\dfrac{1}{2},1\right), \, \mbf{P}_2 = \left(1,\dfrac{1}{2}\right),\, \mbf{P}_3 = \left(\dfrac{1}{4},-1\right), \, \mbf{P}_4 = \left(1,-1\right), \\
    \mbf{P}_5 &= \left(\dfrac{4}{3},-\dfrac{1}{2}\right), \mbf{P}_6 = \left(\dfrac{3}{2},-\dfrac{2}{3}\right), \mbf{P}_7 = \left(\dfrac{5}{4},0\right), \mbf{P}_8 = \left(2,\dfrac{1}{2}\right), \mbf{P}_9 = \left(\dfrac{5}{2},-\dfrac{2}{3}\right),
    \end{align*}
    and weights
    \begin{equation*}
        w_0=1,\, w_1=1, \, w_2 = 2, \, w_3 = 3, \, w_4 = 1, \, w_5=3,\, w_6 = 5, w_7 = 1, \, w_8 = \dfrac{9}{2}, w_9 =1,
    \end{equation*}
    it is possible to define a quartic NURBS curve
    \begin{equation*}
        \phi(u) = \sum_{i=0}^8 w_i \mbf{P}_i N_{i,4}(u).
    \end{equation*}

     \begin{figure}[htbp]
     \centering
  \begin{subfigure}[t]{0.4\textwidth}
    \centering
    \includegraphics[width=\linewidth]{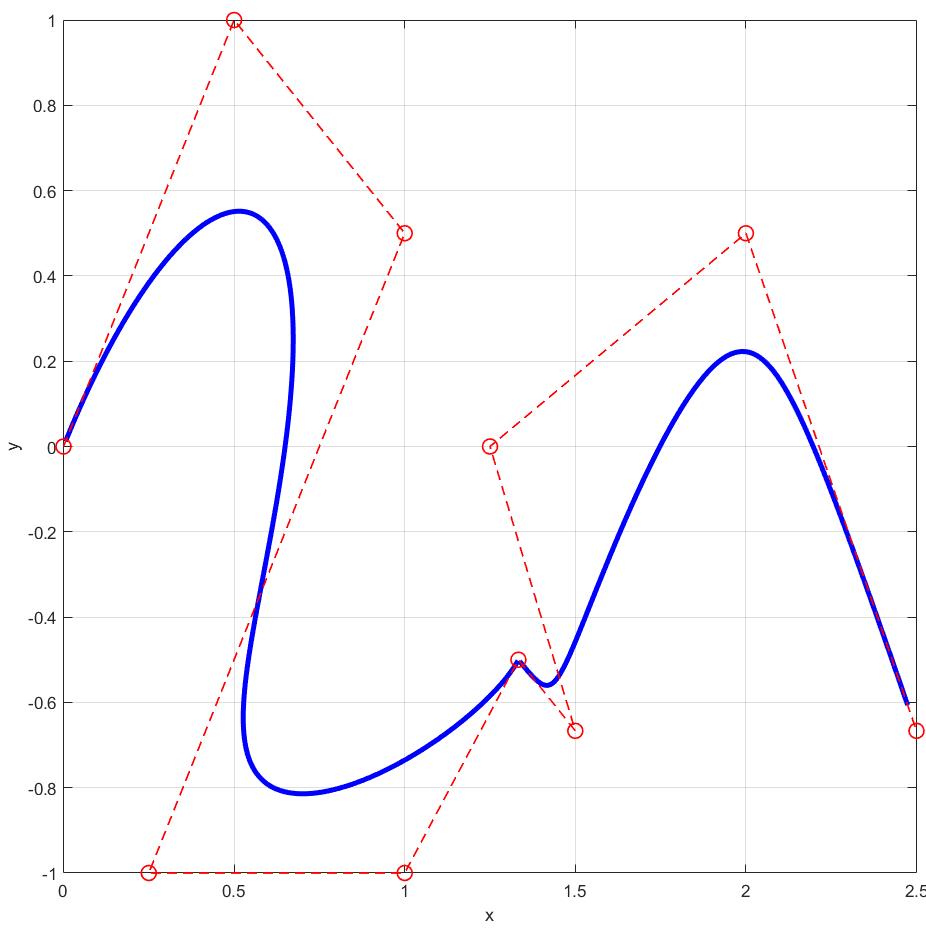}
    \caption{}
  \end{subfigure}
  \begin{subfigure}[t]{0.4\textwidth}
    \centering
    \includegraphics[width=.88\linewidth]{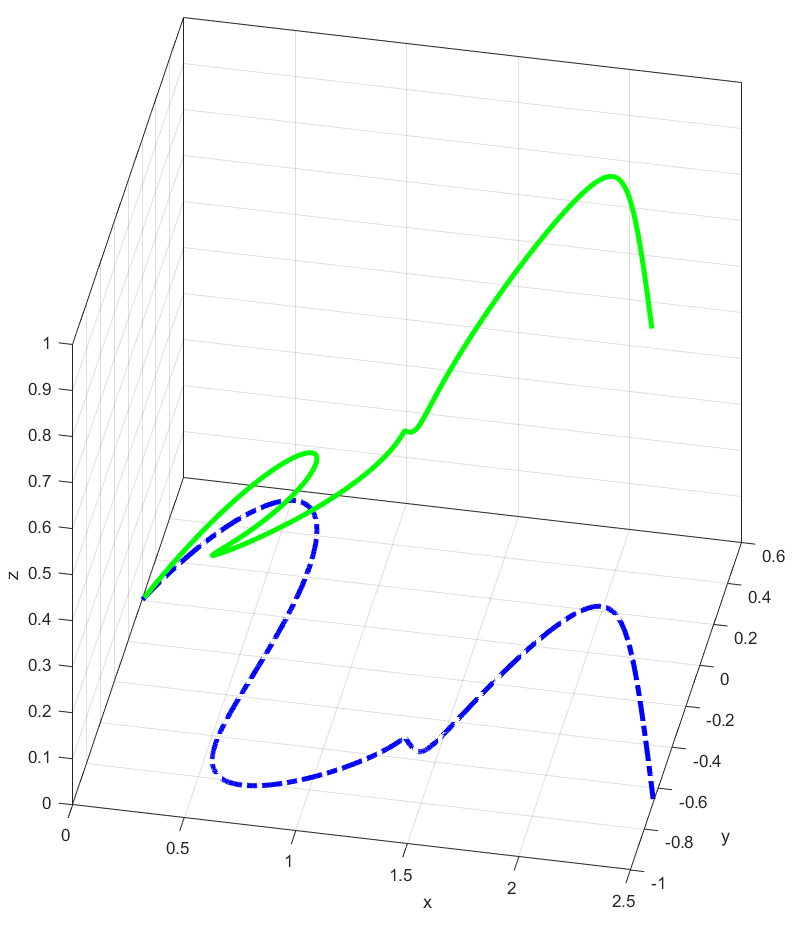}
    \caption{}
  \end{subfigure}
  \caption{NURBS curve of degree $d=4$ with associated control net (a) and the graph of its inverse (b). In dashed line is the input curve $\phi$. Note the $C^0$ continuity in correspondence of $\mbf{P}_5$.}
  \label{quartic_and_inverse}
\end{figure}

  \begin{figure}[h!]
    \centering
    \begin{subfigure}[b]{0.28\textwidth}
        \centering
        \includegraphics[width=\textwidth]{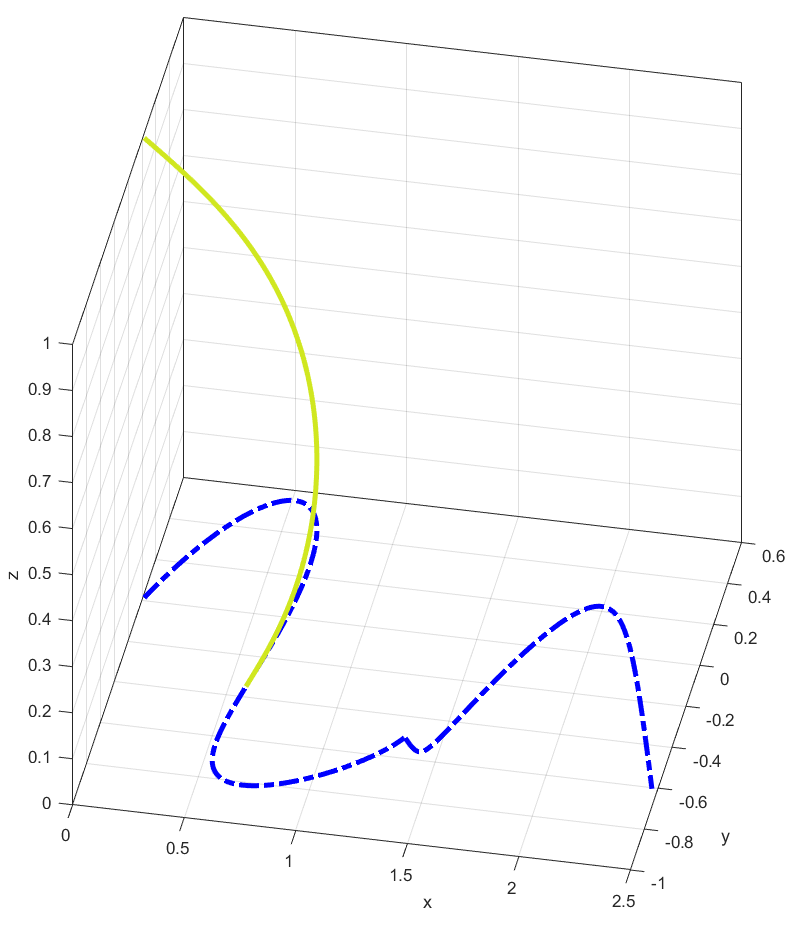}
        \caption{}
    \end{subfigure}
    \hfill
    \begin{subfigure}[b]{0.28\textwidth}
        \centering
        \includegraphics[width=\textwidth]{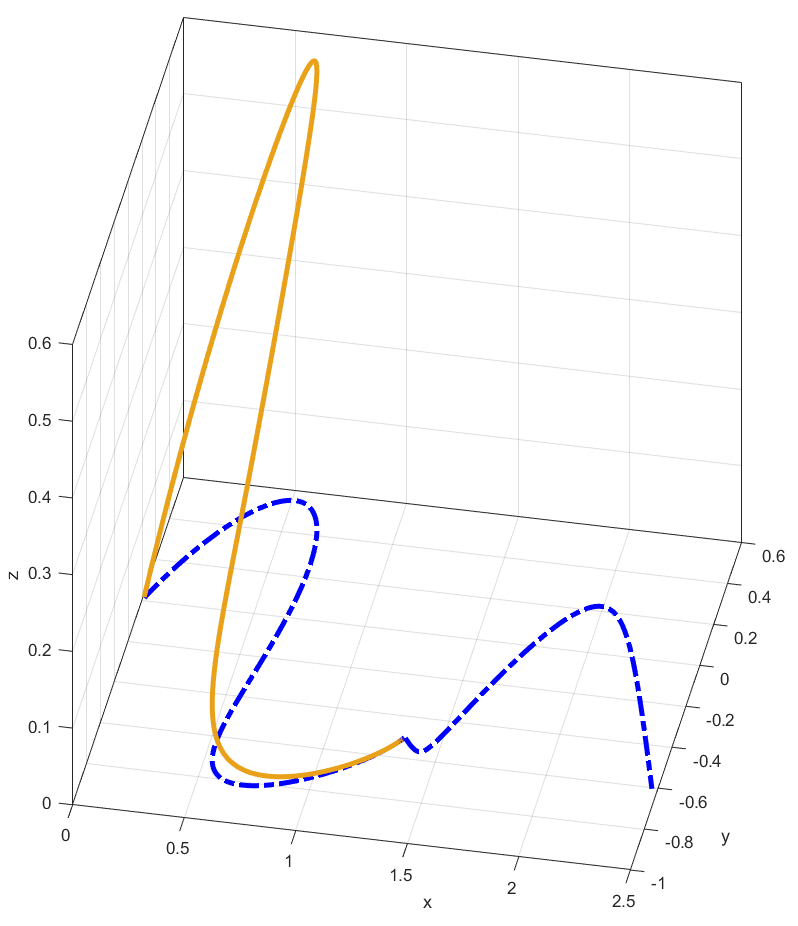}
        \caption{}
    \end{subfigure}
    \hfill
    \begin{subfigure}[b]{0.28\textwidth}
        \centering
        \includegraphics[width=\textwidth]{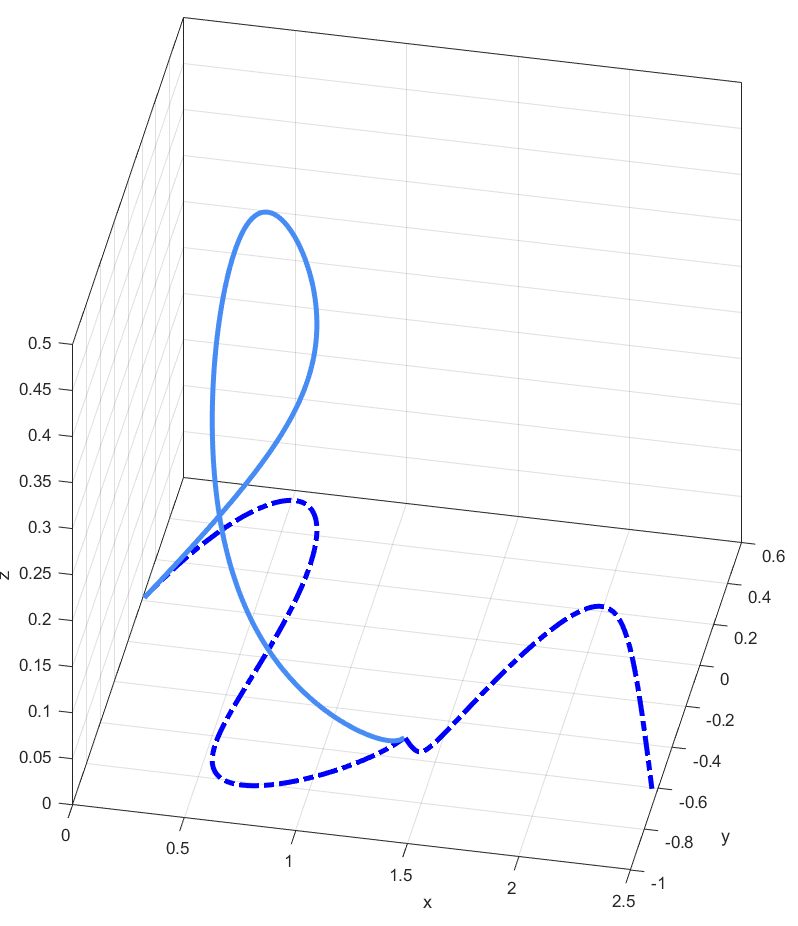}
        \caption{}
    \end{subfigure}

     \begin{subfigure}[b]{0.28\textwidth}
        \centering
        \includegraphics[width=\textwidth]{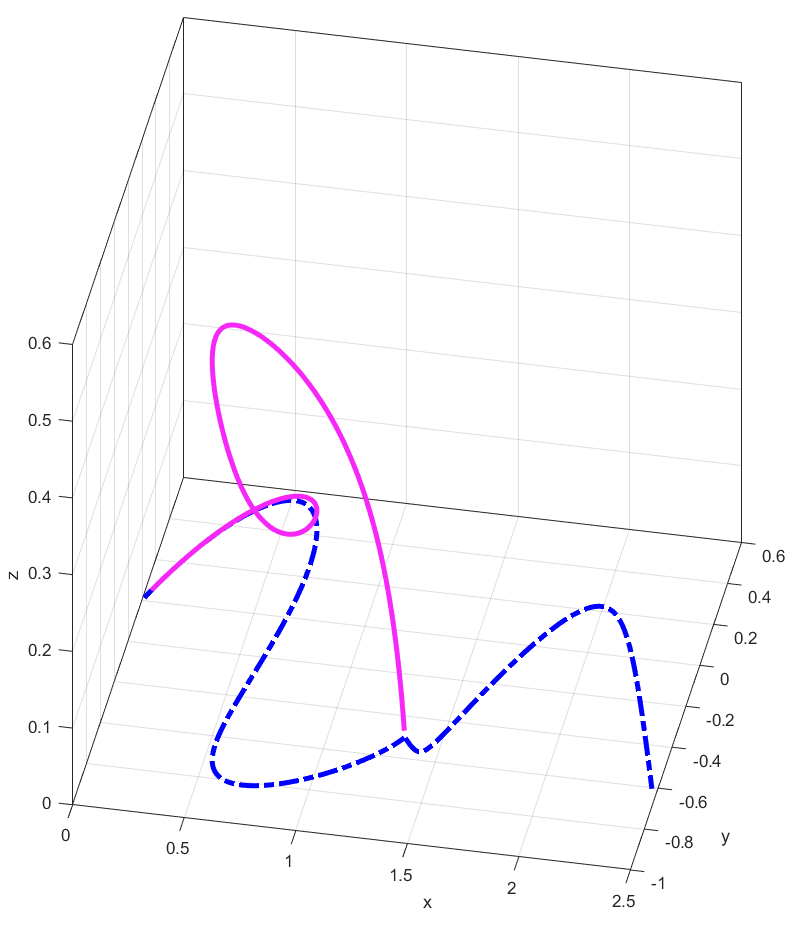}
        \caption{}
    \end{subfigure}
    \hfill
    \begin{subfigure}[b]{0.28\textwidth}
        \centering
        \includegraphics[width=\textwidth]{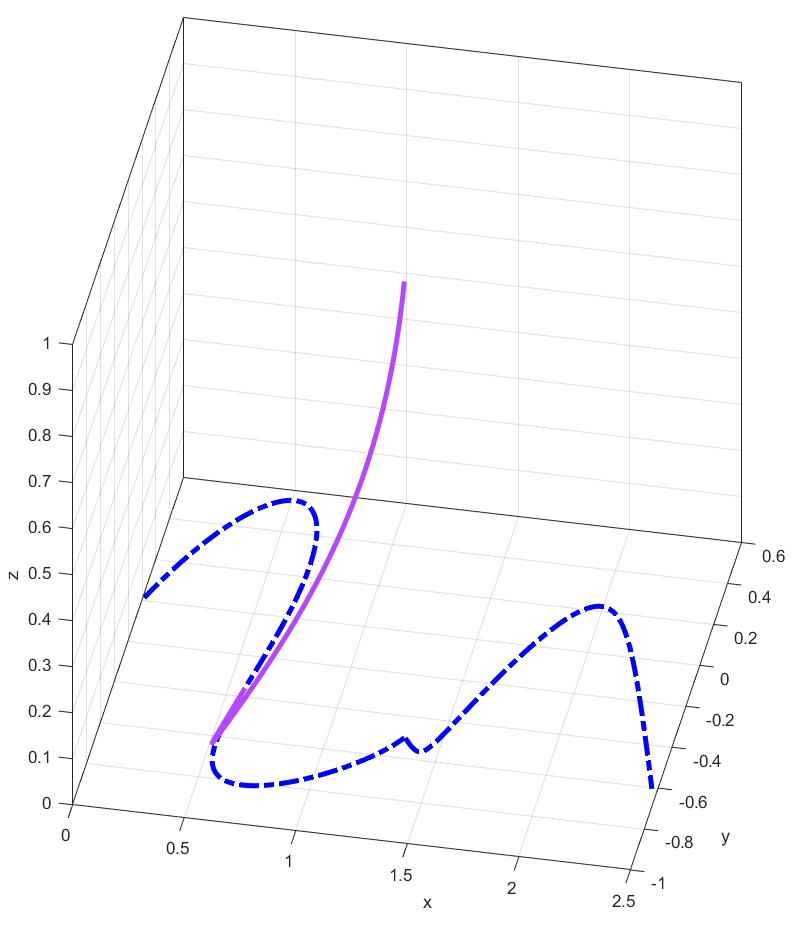}
        \caption{}
    \end{subfigure}
    \hfill
    \begin{subfigure}[b]{0.28\textwidth}
        \centering
        \includegraphics[width=\textwidth]{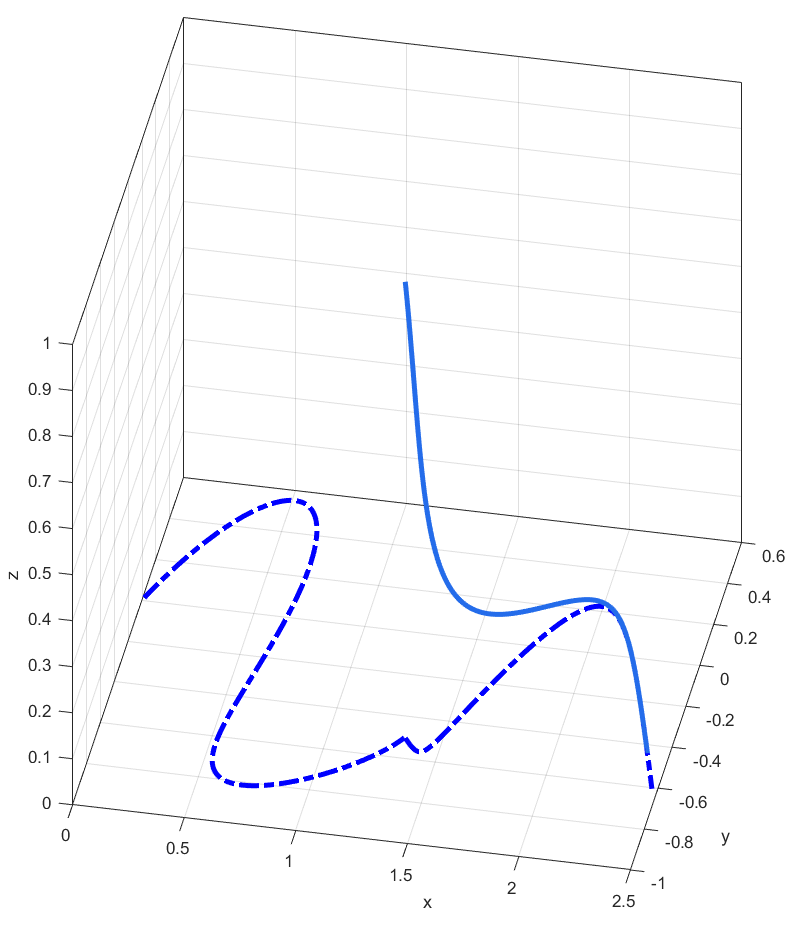}
        \caption{}
    \end{subfigure}

    \begin{subfigure}[b]{0.28\textwidth}
        \centering
        \includegraphics[width=\textwidth]{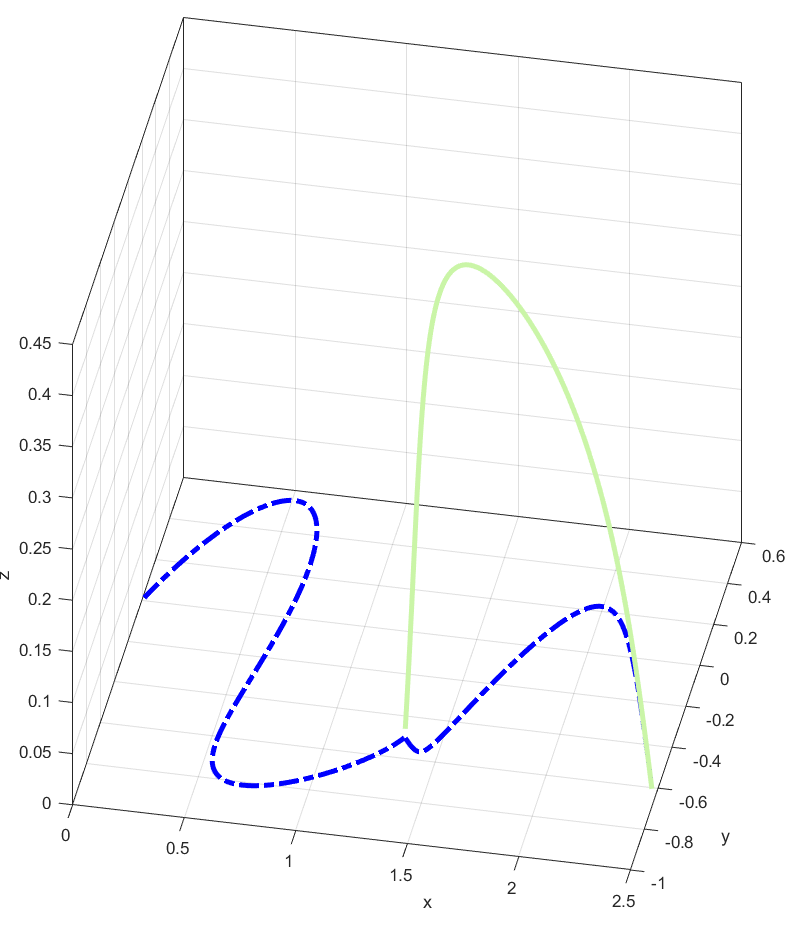}
        \caption{}
    \end{subfigure}
    \hfill
    \begin{subfigure}[b]{0.28\textwidth}
        \centering
        \includegraphics[width=\textwidth]{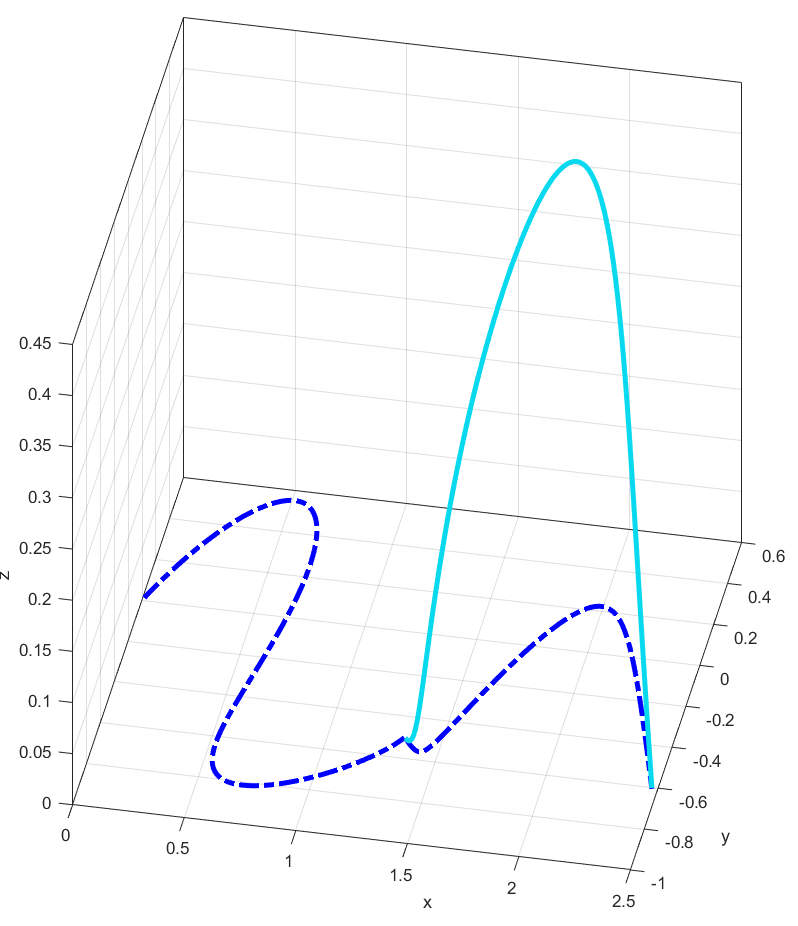}
        \caption{}
    \end{subfigure}
    \hfill
    \begin{subfigure}[b]{0.28\textwidth}
        \centering
        \includegraphics[width=\textwidth]{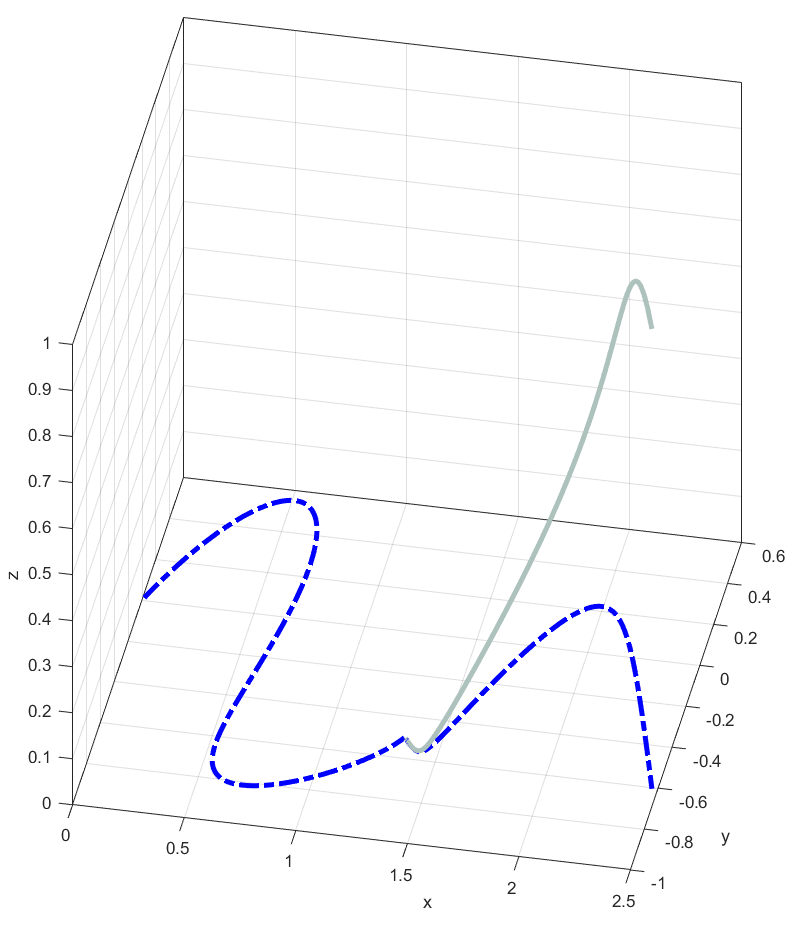}
        \caption{}
    \end{subfigure}

    \caption{From (a) to (i): the set of cubic physical rational splines $\cN_{i,3}, i=0,\dots,8$, defining the inverse function in \eqref{inv_quartic}. In dashed line is the input NURBS curve $\phi$.}
    \label{quartic_NURBS_multiple}
    \end{figure}

    Note that, since $\mu(2/3)=4$, it results $\phi(2/3)=\mbf{P}_5$ as \Cref{quartic_and_inverse}-(a) shows. If we set $p=3$, some of the cubic physical rational splines $\cN_{i,3}$ defined on the vector
    $$\cU=\left\{ \mbf{U}_0, \mbf{U}_0, \mbf{U}_0,\mbf{U}_0,\mbf{U}_1,\mbf{U}_2, \mbf{U}_2,\mbf{U}_2,\mbf{U}_2,\mbf{U}_3,\mbf{U}_3,\mbf{U}_3,\mbf{U}_3\right\}$$
    with
    $$
    \mbf{U}_0 = \mbf{P}_0, \, \mbf{U}_1 = \left(\dfrac{39}{68},-\frac{6}{17}\right),\, \mbf{U}_2 = \mbf{P}_5, \, \mbf{U}_3 = \mbf{P}_9  ,
    $$
    are discontinuous functions at the point $\mbf{U}_2$ (see \Cref{quartic_NURBS_multiple}-(e)-(f)).
    Nevertheless, the proposed construction is completely general for arbitrary multiplicity of the knots and hence we can construct the inverse NURBS map, whose image is depicted in \Cref{quartic_and_inverse}-(b), as

     \begin{equation}\label{inv_quartic}
        \phi^{-1}(x,y) = \sum_{i=0}^8 \xi_{i,3}\, \cN_{i,3}(x,y).
    \end{equation}   

The discontinuities can be avoiding by choosing a degree $p\geq 4$.

    \subsection{Quintic NURBS with self intersection} 
    \label{sec: example - self-intersection}
    In the last example, we consider a quintic curve that presents a self intersection point. If $\phi$ is defined as
    \begin{equation*}
     \phi(u) = \sum_{i=0}^6 w_i \mbf{P}_i N_{i,5}(u),   
    \end{equation*} 
     with control points
     \begin{align*}
    \mbf{P}_0 = \left(0,0\right),\,  \mbf{P}_1 = &\left(1,-\dfrac{1}{3}\right), \, \mbf{P}_2 = \left(3,0\right),\, \mbf{P}_3 = \left(2,\dfrac{1}{3}\right), \, \mbf{P}_4 = \left(1,0 \right), \, \mbf{P}_5 = \left(3,-\dfrac{1}{3}\right), \mbf{P}_6 = \left(4,0\right),
    \end{align*}
    weights
    \begin{equation*}
    w_0=1,\, w_1=3, \, w_2 = 7, \, w_3 = 5, \, w_4 = 7, \, w_5=3,\, w_6 = 1,
    \end{equation*}
    and knot vector $U = \left\{0, 0, 0, 0, 0, 0, \dfrac{1}{2}, 1, 1, 1, 1, 1, 1\right\}$,
    the equation $\phi(u_1)=\phi(u_2)$ gives the parameters that correspond to the self intersection point $S$, as shown in \Cref{quintic_and_inverse}-(a). In this case, if $p=4$, the inverse map 
    \begin{equation}\label{inv_quintic}
        \phi^{-1}(x,y) = \sum_{i=0}^5 \xi_{i,4}\, \cN_{i,4}(x,y),
    \end{equation}
    where $\cN_{i,4}$ are the quartic physical rational functions defined on 
    $$\cU=\left\{ \mbf{U}_0, \mbf{U}_0, \mbf{U}_0,\mbf{U}_0,\mbf{U}_0,\mbf{U}_1, \mbf{U}_2, \mbf{U}_2,\mbf{U}_2,\mbf{U}_2,\mbf{U}_2\right\}, \quad \textnormal{with} \quad \mbf{U}_0 = \mbf{P}_0,\, \mbf{U}_1 =\left(2, \dfrac{2}{23}\right) , \, \mbf{U}_2 = \mbf{P}_6,\vspace{-0.25cm}$$ is still definable, but the value of $\phi^{-1}(S)$ won't be unique and without any a priori information is it not possible to chose among the parameters $u_1$ and $u_2$. \Cref{quintic_and_inverse}-(b) presents the image of the inverse NURBS~map~and~\Cref{Quintic_NURBS_selfint} the graph of the physical B-splines involved in its definition.

       \begin{figure}[htbp]
     \centering
  \begin{subfigure}[t]{0.4\textwidth}
    \centering
    \includegraphics[width=\linewidth]{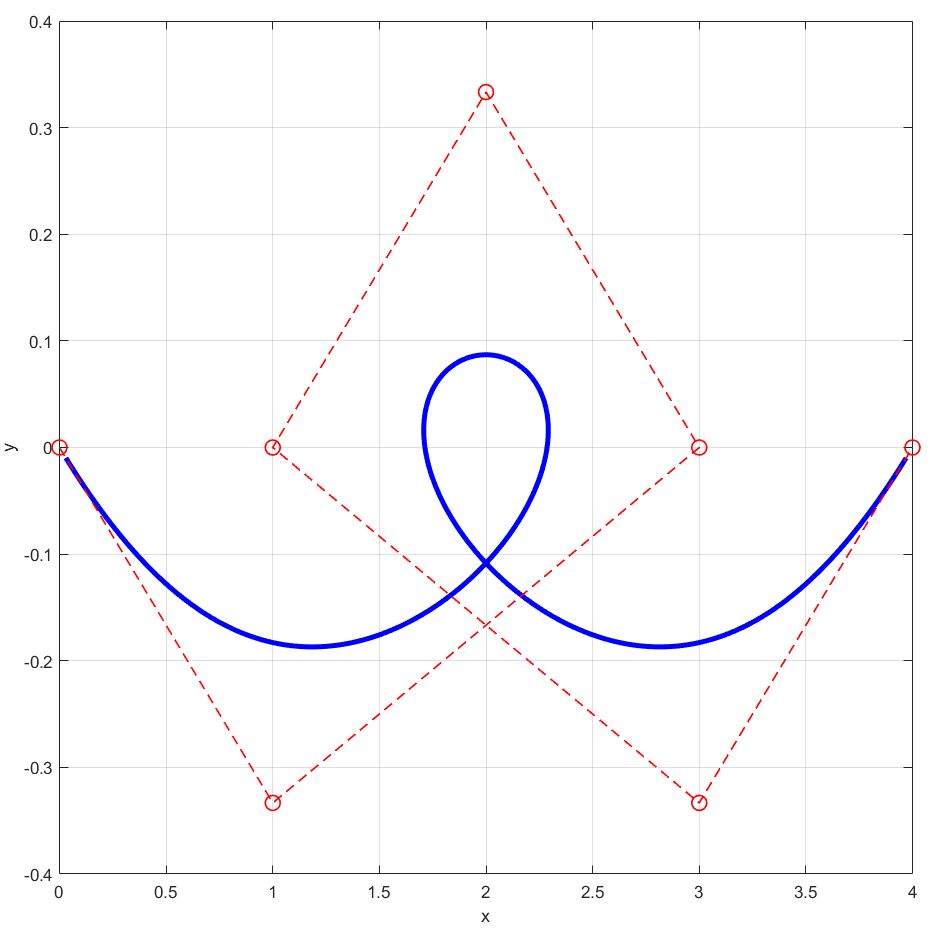}
    \caption{}
  \end{subfigure}
  \begin{subfigure}[t]{0.4\textwidth}
    \centering
    \includegraphics[width=1.17\linewidth]{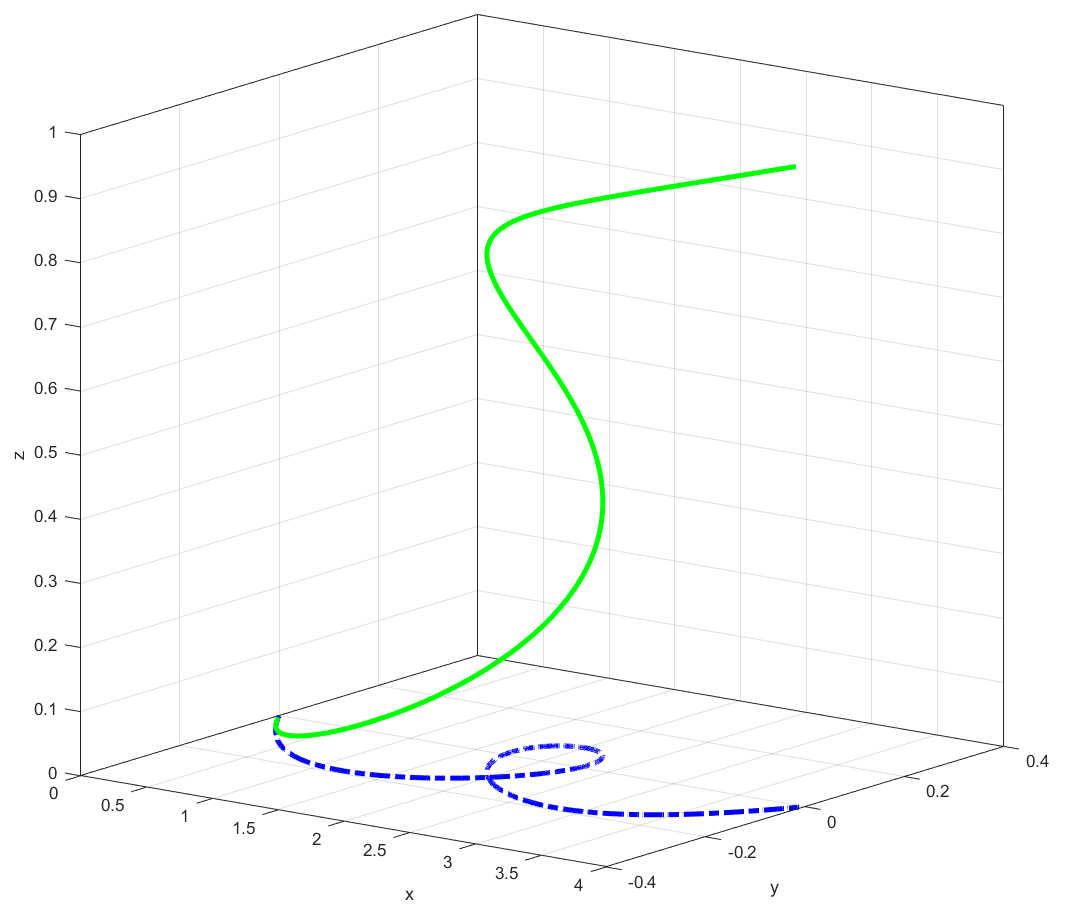}
    \caption{}
  \end{subfigure}
  \caption{Quintic NURBS curve with a self insersection point and its associated control net (a) and the graph of its inverse (b). In dashed line is the input curve $\phi$.}
  \label{quintic_and_inverse}
\end{figure}

 \begin{figure}[htpb]
    \centering
    \begin{subfigure}[b]{0.28\textwidth}
        \centering
        \includegraphics[width=\textwidth]{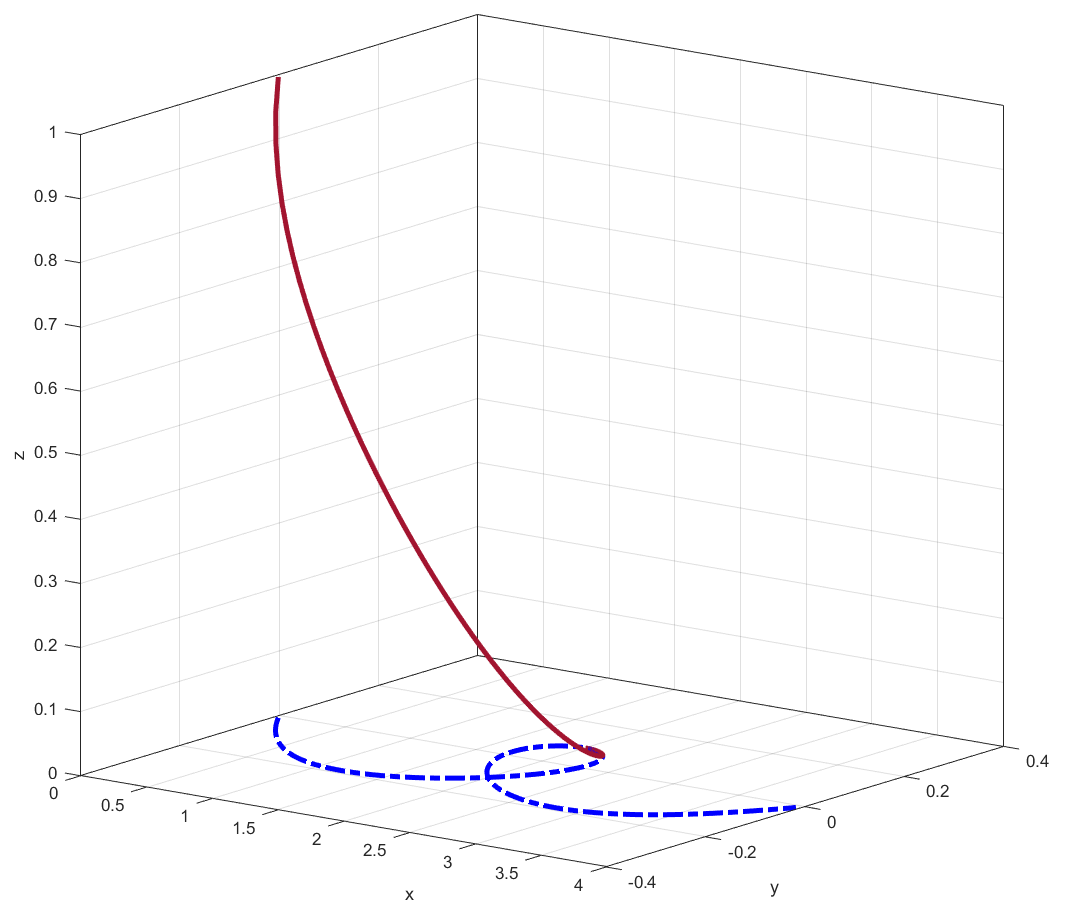}
        \caption{}
    \end{subfigure}
    \hfill
    \begin{subfigure}[b]{0.28\textwidth}
        \centering
        \includegraphics[width=\textwidth]{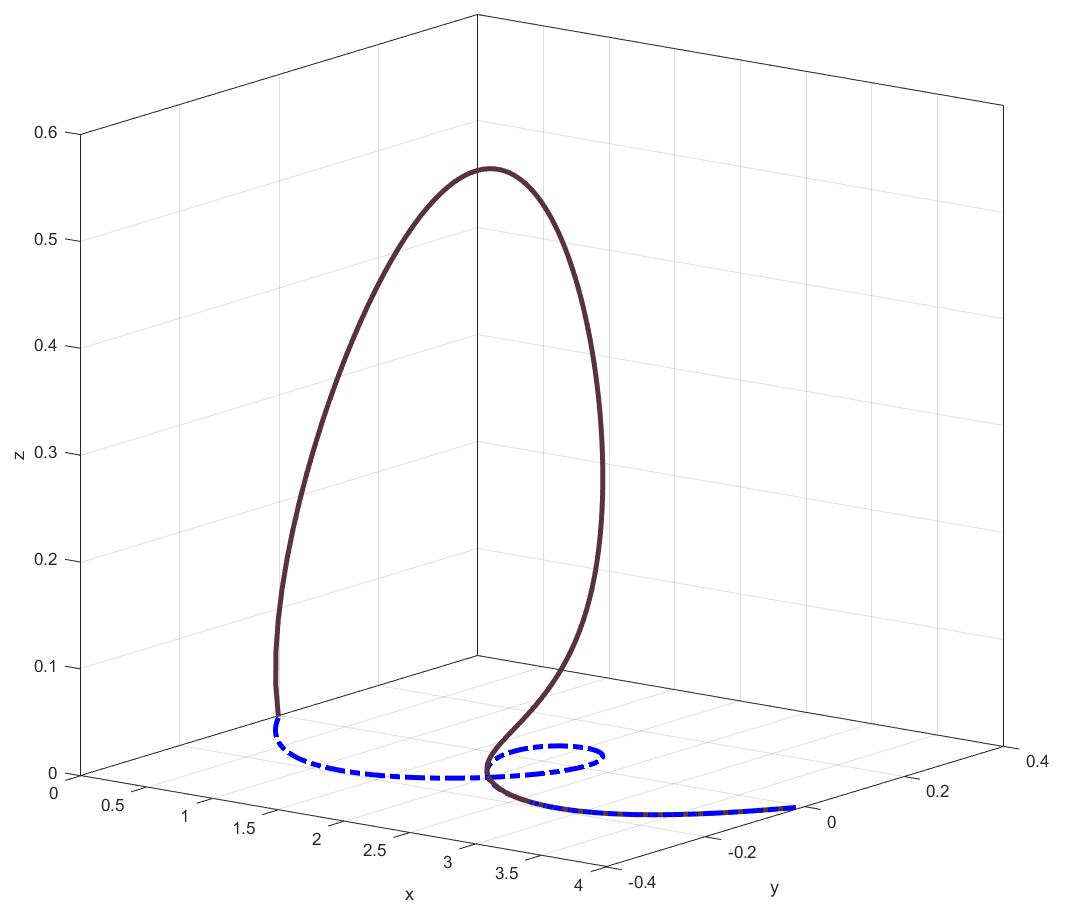}
        \caption{}
    \end{subfigure}
    \hfill
    \begin{subfigure}[b]{0.28\textwidth}
        \centering
        \includegraphics[width=\textwidth]{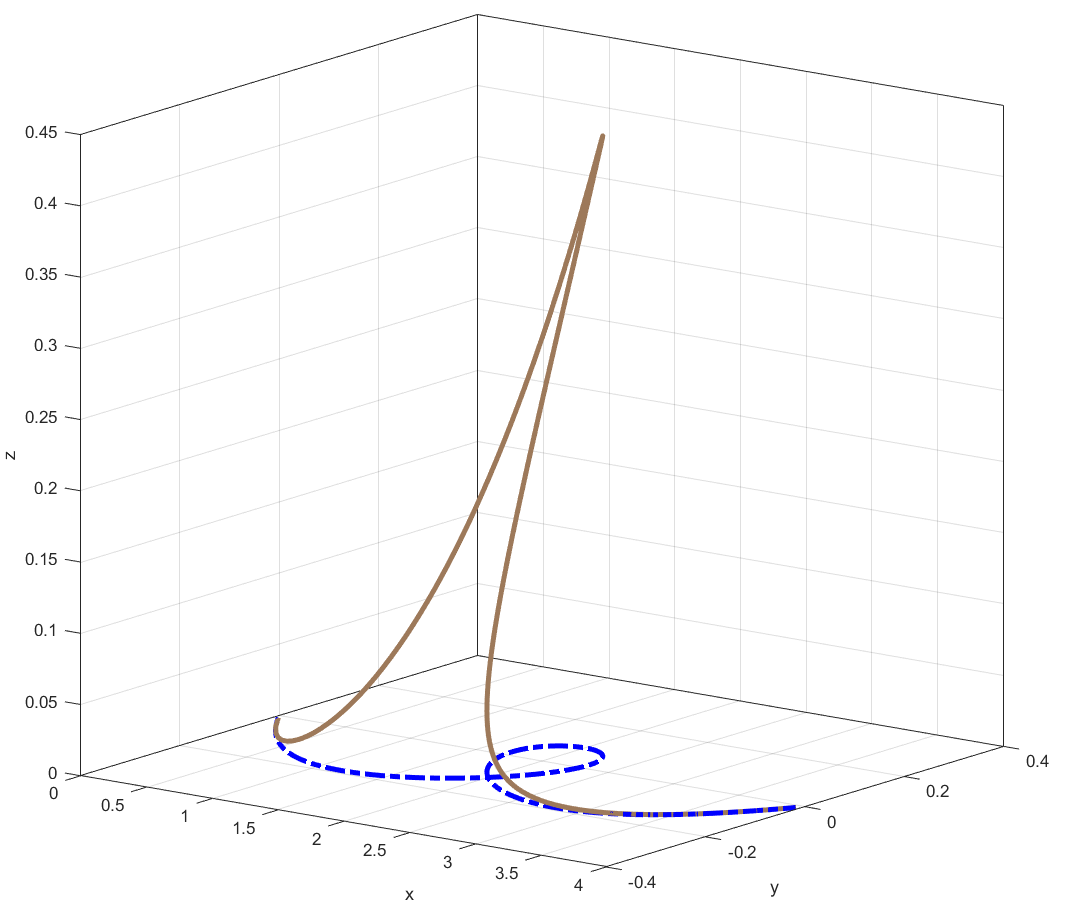}
        \caption{}
    \end{subfigure}

     \begin{subfigure}[b]{0.28\textwidth}
        \centering
        \includegraphics[width=\textwidth]{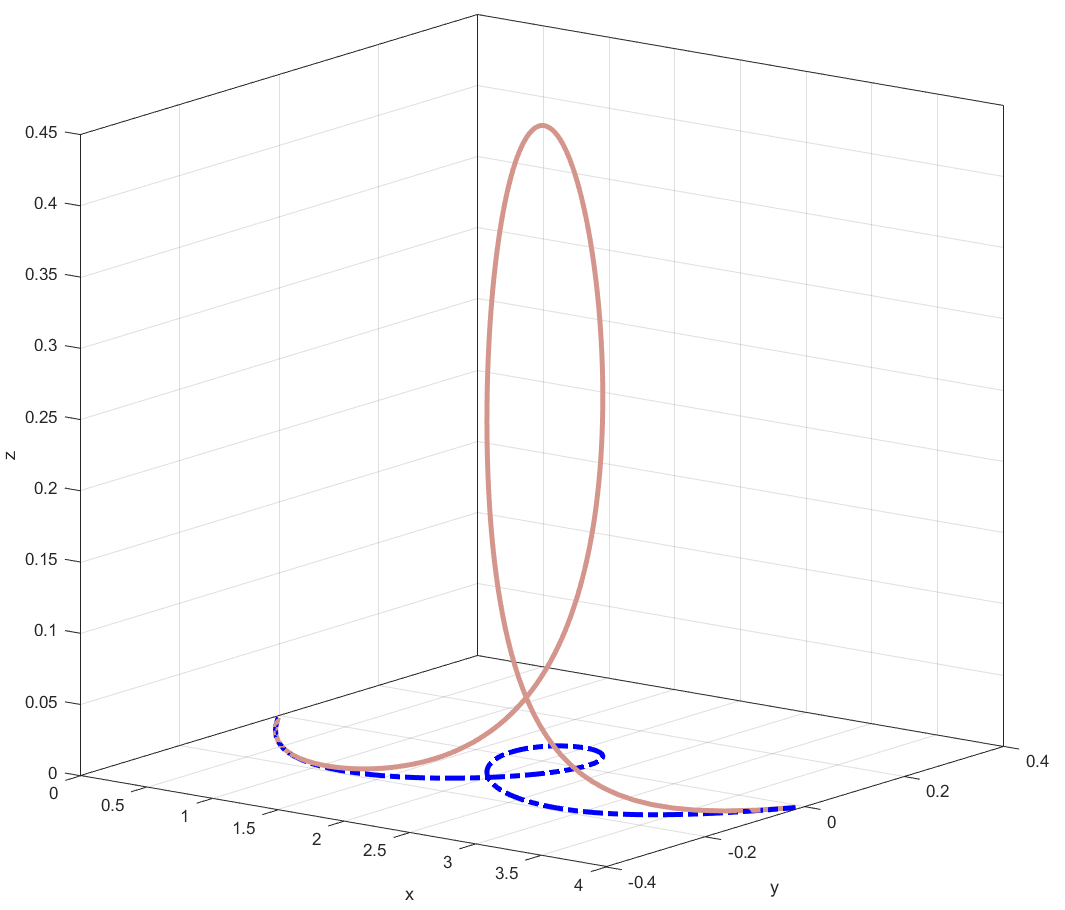}
        \caption{}
    \end{subfigure}
    \hfill
    \begin{subfigure}[b]{0.28\textwidth}
        \centering
        \includegraphics[width=\textwidth]{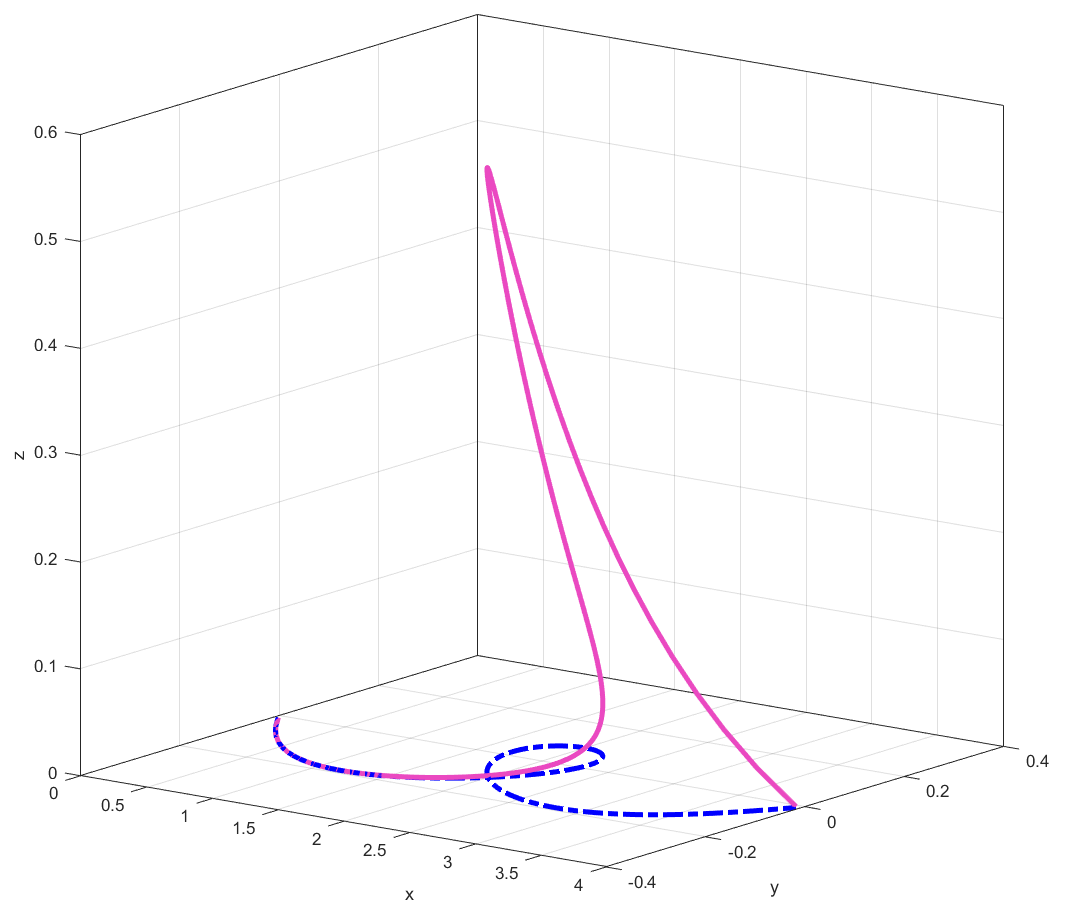}
        \caption{}
    \end{subfigure}
    \hfill
    \begin{subfigure}[b]{0.28\textwidth}
        \centering
        \includegraphics[width=\textwidth]{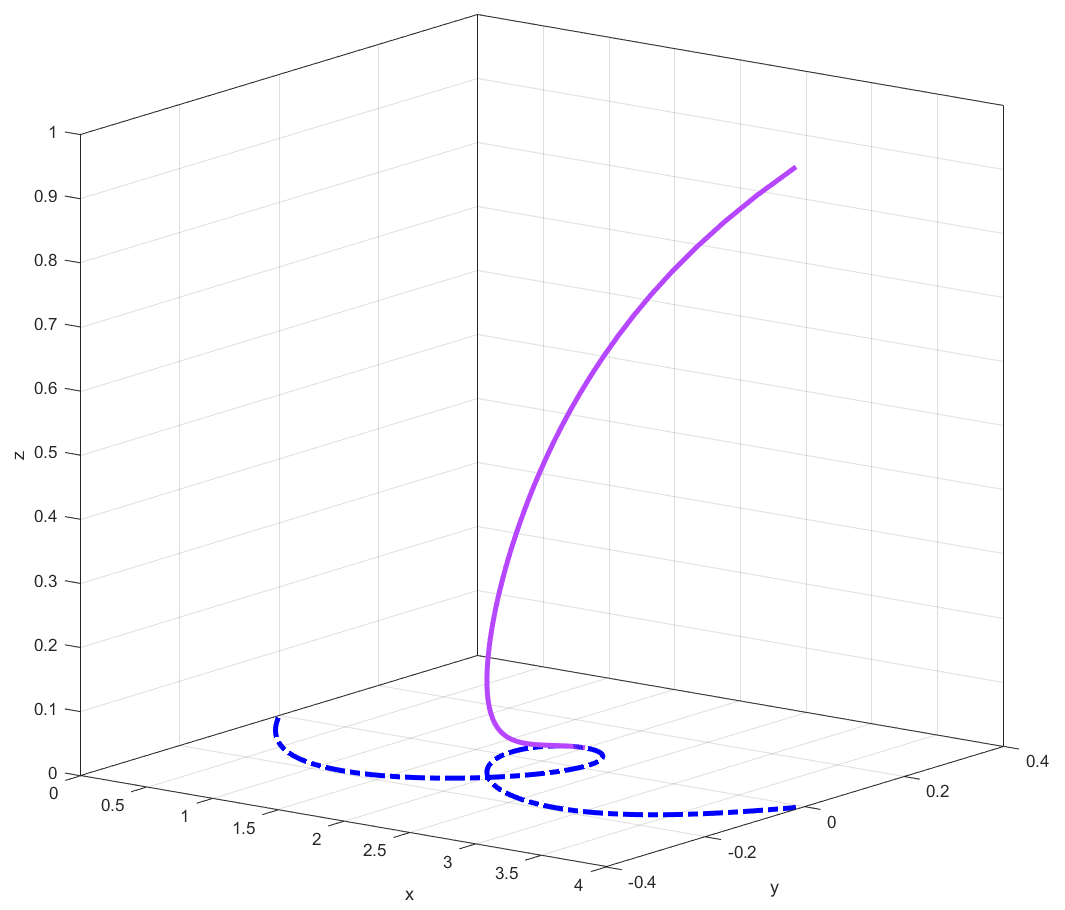}
        \caption{}
    \end{subfigure}

    \caption{From (a) to (f): the quartic physical rational splines $\cN_{i,4}, i=0,\dots,5$, defining the inverse function in \eqref{inv_quintic}. In dashed line is the input NURBS curve $\phi$.}
    \label{Quintic_NURBS_selfint}
    \end{figure}

    \section{Conclusions and Perspectives}
    In this paper we provided an explicit construction of the inverse parametrization of a planar NURBS curve. We first define the inverse by mean of local inverse functions, that are used to define physical rational splines. Therefore we provide a more elegant representation of the inverse as linear combination of such physical rational functions.
    
    This work constitutes a first step toward the development of effective methods for inverting tensor-product NURBS parametrizations of surfaces and volumes, which are required in many problems arising in applied fields such as isogeometric analysis.
    For rational spline parametrization $\phi: [0,1]^2 \xrightarrow{} \R^3$, the map is locally birational under sufficiently generic assumptions, and formulas analogous to those derived in the present paper can be expected.
    By contrast, for patch parametrizations $\phi:[0,1]^n \xrightarrow{} \R^n$, with $n=2,3$, birationality is not a generic property on regions of the domain where the parametrization is rational.
    Indeed, the study of birationality for rational maps in this setting is considerably more challenging.
    Future work will focus on the invertibility of planar NURBS parametrizations $\phi: [0,1]^2 \xrightarrow{} \R^2$ defined by low-bidegree B-spline functions.
    

    \section{Acknowledgements}
    MM is a member of the INdAM research group GNCS,
Italy. 
MM acknowledges the support of the Italian Ministry of University
and Research (MUR) through the PRIN project NOTES (No.
P2022NC97R), funded by the European Union—Next Generation EU. MM is also partially supported by the INdAM GNCS 2025 project ``PASTRAMI - sPline And Solver innovaTions foR Adaptive isogeoMetric analysIs'' (CUP E53C24001950001).
PM is partially supported by the PRIN 2022 grant agreement 40104520.

   \bibliography{ref_invertible_NURBS.bib}
    \end{document}